\documentclass[letterpaper]{article}

\usepackage{template}

\usepackage[utf8]{inputenc} 
\usepackage[T1]{fontenc}    
\usepackage{hyperref}       
\usepackage{url}            
\usepackage{booktabs}       
\usepackage{amsfonts}       
\usepackage{nicefrac}       
\usepackage{microtype}      
\usepackage{xcolor}         

\hypersetup{colorlinks=true, citecolor=cyan!50!black, linkcolor=red!70!black}

\title{Differentially Private Approval-Based Committee Voting}

\author{%
  Zhechen Li$^1$, Zimai Guo$^1$, Lirong Xia$^3$, Yongzhi Cao$^{1*}$, Hanpin Wang$^{1,4}$ \\
  $^1$Peking University, $^{2,3}$Rensselaer Polytechnic Institute, $^4$Guangzhou University \\
  $^1$\texttt{\{lizhechen,endergzm,caoyz,whpxhy\}@pku.edu.cn} \\
  $^3$\texttt{xialirong@gmail.com} \\
}

\usepackage{fancyhdr}
\usepackage{amsmath,amsfonts,amssymb,amsthm}
\usepackage{booktabs}
\usepackage{pifont}
\usepackage[ruled,linesnumbered,vlined]{algorithm2e}
\usepackage{mathrsfs}
\usepackage{multirow,multicol}
\usepackage{subfigure}
\usepackage{enumerate}
\usepackage{bbm}
\usepackage{makecell}
\usepackage{subfigure}
\usepackage{graphicx}
\usepackage{float}
\usepackage{tikz}
\usepackage{tikz-cd}
\usepackage{bm}
\usepackage{natbib}

\newtheorem{theorem}{Theorem}
\newtheorem{definition}{Definition}
\newtheorem{proposition}{Proposition}
\newtheorem{lemma}{Lemma}
\newtheorem{corollary}{Corollary}

\makeatletter
\providecommand{\sketchname}{\hspace{\parindent}\textsc{\rm Proof Sketch}}
\newenvironment{pfsketch}[1][\sketchname]{\par
  \pushQED{\qed}%
  \normalfont \topsep6\p@\@plus6\p@\relax
  \trivlist
  \item[\hskip\labelsep
        \itshape
    #1\@addpunct{.}]\ignorespaces
}{\popQED\endtrivlist\@endpefalse}
\makeatother

\newcommand{\powerset}[1]{\mathcal{P}(#1)}

\def\JR{\text{JR}}
\def\PJR{\text{PJR}}
\def\EJR{\text{EJR}}

\def\supp{\operatorname{Supp}}
\def\Ak{\powerset{A,k}}
\def\themap{\powerset{A}^n\times\N\to \mathcal{R}(\powerset{A})}

\def\p{\mathbb{P}}
\def\N{\mathbb{N}_+}
\def\R{\mathbb{R}}

\def\erm{ {\rm e} }

\begin{document}

\title{Differentially Private Approval-Based Committee Voting}

\maketitle

\begin{abstract}
  In this paper, we investigate tradeoffs among differential privacy (DP) and several representative axioms for approval-based committee voting, including justified representation, proportional justified representation, extended justified representation, Pareto efficiency, and Condorcet criterion. Without surprise, we demonstrate that all of these axioms are incompatible with DP, and thus establish both upper and lower bounds for their two-way tradeoffs with DP. Furthermore, we provide upper and lower bounds for three-way tradeoffs among DP and every pairwise combination of such axioms, revealing that although these axioms are compatible without DP, their optimal levels under DP cannot be simultaneously achieved. Our results quantify the effect of DP on the satisfaction and compatibility of the axioms in approval-based committee voting, which can provide insights for designing voting rules that possess both privacy and axiomatic properties.
\end{abstract}

\section{Introduction}

Voting is widely used in elections, where a group of voters cast ballots representing their preferences over a set of alternatives, and a voting rule is subsequently employed to determine the winner. In this paper, our primary focus lies on committee voting rules, where the winner is not a singular alternative but a group of alternatives. Specifically, we focus on approval-based committee voting rules, abbreviated as {\em ABC rules}. In this framework, each voter's preference is dichotomous, i.e., they categorize the set of alternatives into approved and disapproved subsets. Similar to single-winner voting, the evaluation of ABC rules are usually based on various axiomatic properties \citep{Plott76:Axiomatic}, known as voting axioms. From the perspective of axiomatic approaches, the primary aspect that sets committee voting apart from single-winner voting is to consider fairness, which can be described by proportionality. Intuitively, proportionality requires the ABC rule to grant smaller and larger groups of voters a fair consideration of their preferences \citep{lackner2023multi}, and has been extensively studied in social choice thoery \citep{aziz2017justified,peters2020proportionality,skowron2021proportionality}, participatory budgeting~\citep{brill2023proportionality,aziz2024fair}, and machine learning \citep{chen2019proportionally, micha2020proportionally}.

In recent years, privacy concerns in voting have emerged as a pressing public issue. For real-world elections, minimizing the amount of privacy leakage helps against censorship, coercion, and bribery~\citep{ao2020private}. However, revealing the winners of elections can sometimes cause privacy leakage. To illustrate this, just consider the following scenario, which is similar to the motivating example in~\citep{ao2020private}. In an approval-based committee election, an adversary discovers (by questionnaries, social media analysis, or other methods) that the preferences of all the voters except Alice form a tie. Then the adversary can obtain the preference of Alice accurately only by observing the winner of the election, even if he has no prior knowledge about Alice, since Alice must vote for the winning committee to break tie. From the perspective of differential privacy~\citep{Dwork06D}, the occurrence of the scenario above should be regarded as a violation of privacy.

Till now, a body of literature has focused on applying differential privacy in (single-winner) voting rules~\citep{shang2014application,hay2017differentially,yan2020private}. These works primarily involve the application of various randomized mechanisms to some existing voting rules, evaluating the utility loss (usually measured by accuracy or mean square error) under some privacy guarantees, but overlook the satisfaction of the voting axioms. \citet{DBLP:conf/ijcai/Lee15} discussed the strategyproofness under DP, and~\citet{li2022differentially,li2023trading} explored the tradeoff between DP and several voting axioms, including Condorcet criterion, Pareto efficiency, SD-strategyproofness, etc.
However, in committee voting, the tradeoff between privacy and voting axioms is still unclear. Therefore, the following question remains largely open.

\begin{center}
  \em
  What is the tradeoff among DP and voting axioms for approval-based committee voting rules?
\end{center}

\begin{figure*}[h]
  \centering
  \includegraphics{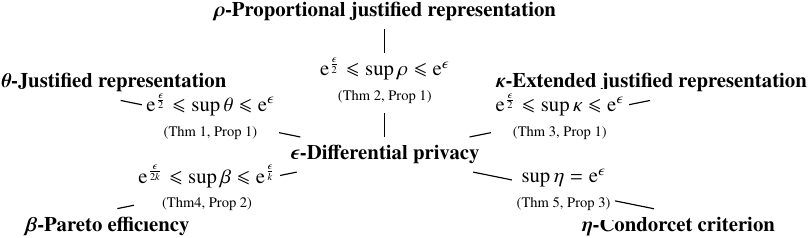}
  \caption{Two-way tradeoffs between DP and axioms.}
  \label{fig: 2-way tradeoff}
\end{figure*}
\begin{figure*}[h]
  \centering
  \includegraphics{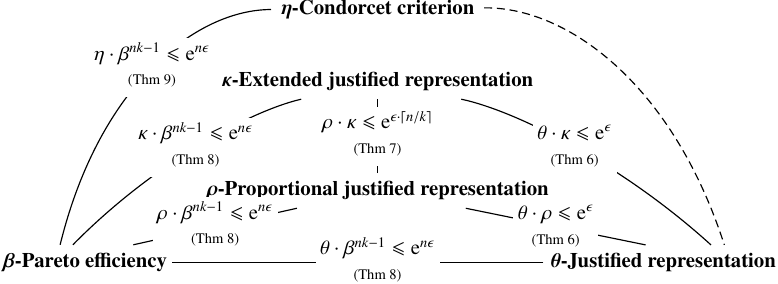}
  \caption{Three-way tradeoffs among DP and axioms (the dash line indicates that standard Condorcet criterion and justified representation are incompatible, so there is no 3-way tradeoff for them).}
  \label{fig: 3-tradeoff}
\end{figure*}

\subsection{Our contributions}

Our conceptual contributions encompass defining approximate versions of several axioms for randomized ABC rules, including the justified representation, proportional justified representation, extended justified representation, Pareto efficiency, and Condorcet criterion (Definitions \ref{def: approx-JR}--\ref{def: approx-CC}).

Our theoretical contributions are two-fold. Firstly, we investigate the 2-way tradeoff between DP and a single axiom. We provide tradeoff theorems (with both upper and lower bounds) between DP and approximate justified representation, proportional justified representation, extended justified representation, and Pareto efficiency (Theorems \ref{thm: 2-JR-upper}--\ref{thm: 2-PE-upper}, Propositions \ref{prop: 2-JR-lower}--\ref{prop: 2-PE-lower}). For Condorcet criterion, we provide a tight bound for its tradeoff with DP (Theorem \ref{thm: 2-CC-upper}, Proposition \ref{prop: 2-CC-lower}). We summarize this part of contributions in Figure \ref{fig: 2-way tradeoff}, where the expressions displayed on the lines illustrate the bounds of the approximate axioms under $\epsilon$-DP.

Secondly, we explore the 3-way tradeoffs between DP and pairwise combinations of the axioms. Since justified representation and Condorcet criterion are incompatible even without considering DP, we ignore this pair of axioms in 3-way tradeoffs. For the other pairwise combinations among the aforementioned axioms, we investigate their tradeoffs under DP by proving their upper bounds (Theorems \ref{thm: 3-JR-PJR-upper}--\ref{thm: 3-PE-CC-upper}). We show that their exists a tradeoff between each pair of the axioms, though they are compatible without DP. This part of results is summarized in Figure \ref{fig: 3-tradeoff}, where the expressions displayed on the lines illustrate the upper bounds of 3-way tradeoffs between axioms under $\epsilon$-DP.

Due to the space limit, we omit some proofs in the paper, which can be found in Appendix.

\subsection{Related work and discussions}

In the context of social choice, there is a large literature on approval-based committee voting~\citep{Faliszewski2017:Multiwinner,lackner2023multi}. Most of them focused on examining the satisfaction of standard forms of axiomatic properties, such as proportionality \citep{aziz2017justified,brill2023robust}, Pareto efficiency \citep{lackner2020utilitarian,aziz2020computing}, Condorcet criterion \citep{darmann2013hard}, monotonicity \citep{elkind2017properties,sanchez2019monotonicity}, and consistency \citep{lackner2018consistent}. For approximate axiomatic properties, \citet{procaccia2010can} discussed how much a strategyproof randomized rule could approximate a deterministic rule. \citet{birrell2011approximately} explored the approximate strategyproofness for randomized voting rules. \citet{xia2021semi} studied the semi-random satisfaction of voting axioms. For committee voting rules, \citet{jiang2020approximately}, \citet{munagala2022approximate}, and \citet{cheng2020group} studied the approximations of proportionality.

To the best of our knowledge, \citet{shang2014application} first explored the application of DP to rank aggregations (a generalized problem of voting) and derived upper bounds on error rate under DP.
In a similar vein, \citet{DBLP:conf/ijcai/Lee15} proposed a tournament voting rule that achieves both DP and approximate strategyproofness.
\citet{hay2017differentially} designed two novel rank aggregation algorithms based on Quicksort and Kemeny-Young method by applying the famous Laplace mechanism and exponential mechanism.
\citet{kohli2018epsilon} investigated DP and strategyproofness for anonymous voting rules on single-peaked preferences.
\citet{torra2019random} analyzed the condition, under which the random dictatorship, a well-known randomized single-wineer voting rule, satisfies DP. Further, their study proposed improvements to achieve it in general cases.
\citet{wang2019aggregating} analyzed the privacy of positional voting and proposed a new noise-adding mechanism that outperforms the naive Laplace mechanism in terms of accuracy.
\citet{yan2020private} addressed the tradeoff between accuracy and local DP in rank aggregation through Laplace mechanism and randomized response.

Most of the aforementioned works lack consideration for the tradeoff between privacy and voting axioms, which has been considered by the following works. \citet{ao2020private} introduced distributional DP~\citep{bassily2013coupled} to voting and studied the privacy level of several commonly used voting rules. \citet{li2022differentially} introduced a novel family of voting rules based on Condorcet method. They show that their rules can satisfy DP along with various voting axioms. Besides, they also explored the tradeoff between DP and several axioms, including Pareto efficiency, Condorcet criterion, etc. \citet{mohsin2022learning} studied how to design and learn voting rules that embed DP and some axiomatic properties. \citet{li2023trading} investigated both 2-way and 3-way tradeoffs among DP and voting axioms. Their analysis shows that DP not only constraints the achievable level single voting axioms, but also exacerbates the incompatibility between them. However, all of these works only focused on single-winner voting. The analysis for DP-axiom tradeoff in committee voting is still needed.

Beyond social choice, DP has also been considered in other topics of economics, such as mechanism design~\citep{pai2013privacy,xiao2013privacy}, auction~\citep{jamshidi2024differentially,jia2023incentivising}, matchings and resource allocation~\citep{hsu2016private,kannan2018private,angel2020private,chen2023differentially}.

\section{Preliminaries}
\label{sec: prelim}

Let $A=\{a_1,a_2,\ldots,a_m\}$ be a finite set of $m\geqslant 3$ distinct alternatives. Let $N = \{1,2,\ldots,n\}$ denote a finite set of $n$ voters. We assume that voters' preferences are approval-based, i.e., voters distinguish between alternatives they approve and those they disapprove. Therefore, each voter $i$'s ballot can be denoted by a (non-empty) set $P_i\subseteq A$, representing the set of alternatives approved by the $i$-th voter. Then the {\em (preference) profile}, i.e., the collection of all voters' preferences is denoted by $P=(P_1,P_2,\ldots,P_n)\in\powerset{A}^n$, where $\powerset{A}$ is the power set of $A$. In addition, we use $P_{-i}$ to denote the profile obtained by removing the $i$-th voter's preference from $P$.

In the paper, we are interested in committees of a specific size $k\in\N$. The input for choosing such a committee is called a {\em voting instance} $E=(P, k)$, consisting of a profile $P$ and a desired committee size $k$. Under the settings above, a (randomized) approval-based committee voting rule (ABC rule for short) is a mapping $f\colon \powerset{A}^n\times \N\to \mathcal{R}(\powerset{A})$, where $\mathcal{R}(\powerset{A})$ denotes the set of all random variables on $\powerset{A}$. For any $W\in\powerset{A}$, the winning probability of $W$ is denoted by $\p[f(P,k)=W]$. Then for any voting instance $(P,k)$, we have $\supp f(P,k)\subseteq \powerset{A,k}$, where $\supp f(P,k)$ denotes the support set of $f(P,k)$ and $\powerset{A,k}=\{W\subseteq A: |W|=k\}$. An ABC rule is neutral if for any voting instance $(P,k)$ and any permutation $\sigma$ on $A$, $\sigma\cdot f(P,k)=f(\sigma\cdot P,k)$.

\paragraph{Differential privacy (DP).} In a word, DP requires a function to return similar output when receiving similar inputs. Here, the similarity between inputs is ensured by the notion of {\em neighboring datasets}, i.e., two datasets differing on at most one record. Formally, DP is defined as follows.

\begin{definition}[\bf\boldmath $\epsilon$-Differential privacy, $\epsilon$-DP for short \citep{Dwork06D}]
    \label{def: dp}
    A mechanism $f:\mathcal{D} \to \mathcal{O}$ satisfies $\epsilon$-DP if for all $O\subseteq \mathcal{O}$ and each pair of neighboring databases $D,D'\in \mathcal{D}$,
    \begin{align*}
        \p[f(D)\in O] \leqslant \erm^{\epsilon}\cdot \p[f(D')\in O].
    \end{align*}
\end{definition}

The probability of the above inequality is taken over the randomness of the mechanism $f$. The smaller $\epsilon$ is, the better privacy guarantee can be offered. In the context of multi-winner voting, the mechanism $f$ in Definition \ref{def: dp} is an ABC rule and its domain $\mathcal{D}=\powerset{A}^n\times \N$. Further, the pair of neighboring databases $D=(P,k),D'=(P',k)$ are two voting instances with the same committee size $k$ and their profiles differ on at most one voter's vote, i.e., $P_{-j}=P_{-j}'$ holds for every voter $j\in N$.

\paragraph{Axioms of approval-based committee voting.} The specific axioms considered in the paper can be roughly divided into two parts, which are listed below, respectively.

The first part of axioms is usually called proportionality in committee voting, including justified representation, proportional justified representation, and extended justified representation. For the sake of simplicity, we call them {\em JR-family} throughout the paper. All of the axioms in JR-family rely on the notion of {\em $\ell$-cohesive group}. Given $\ell\geqslant 1$, a group of voters $V\subseteq N$ is $\ell$-cohesive if it satisfies $|V|\geqslant \ell\cdot \frac{n}{k}$ and $|\bigcap_{i\in V} P_i|\geqslant \ell$. Based on this definition, the axioms in JR-family are listed as follows.

\textbullet~{\bf Justified Representation (JR, \citep{aziz2017justified})}: A committee $W\in \Ak$ satisfies JR for voting instance $E=(P,k)$ if for any $1$-cohesive group $V$, there exists $i\in V$, such that $|P_i\cap W|\geqslant 1$. The family of committees satisfying JR for voting instance $(P,k)$ is denoted by $\JR(P,k)$. An ABC rule $f$ satisfies JR if $\supp f(P,k)\subseteq \JR(P,k)$ holds for any voting instance $(P,k)\in \powerset{A}\times \N$.

\textbullet~{\bf Proportional Justified Representation (PJR, \citep{sanchez2017proportional})}: A committee $W\in \Ak$ satisfies PJR for voting instance $E=(P,k)$ if for any $\ell$-cohesive group $V$, $|W\cap (\bigcup_{i\in V} P_i)|\geqslant \ell$. The family of committees satisfying PJR for voting instance $(P,k)$ is denoted by $\PJR(P,k)$. An ABC rule $f$ satisfies PJR if $\supp f(P,k)\subseteq \JR(P,k)$ holds for any voting instance $(P,k)\in \powerset{A}\times \N$.

\textbullet~{\bf Extended Justified Representation (EJR, \citep{aziz2017justified})}: A committee $W\in \Ak$ satisfies EJR for voting instance $E=(P,k)$ if for any $\ell$-cohesive group $V$, there exists $i\in V$, such that $|P_i\cap W|\geqslant \ell$. The family of committees satisfying EJR for voting instance $(P,k)$ is denoted by $\EJR(P,k)$. An ABC rule $f$ satisfies EJR if $\supp f(P,k)\subseteq \JR(P,k)$ holds for any voting instance $(P,k)\in \powerset{A}\times \N$.

Based on the implication relationship between axioms, there is a hierarchy in JR-family \citep{lackner2023multi}, shown in the following diagram, where $a\to b$ means that $a$ implies $b$, i.e., any rule satisfying $a$ satisfies $b$.
\begin{equation*}
    \begin{aligned}
        {\textstyle \text{({\it Strongest})}}\quad \text{EJR} ~\to~ \text{PJR} ~\to~ \text{JR} \quad{\textstyle \text{({\it Weakest})}}
    \end{aligned}
\end{equation*}

The second part of axioms can be regarded as efficiency notions, including Pareto efficiency and Condorcet criterion, defined as follows. Here, we adopt the definitions in \citep{lackner2023multi}.

\textbullet~ {\bf Pareto Efficiency (PE)}: An ABC rule $f$ satisfies Pareto efficiency if for any voting instance $E=(P,k)$, $f(E)$ is never Pareto dominated by any other committee, i.e., there does not exist committee $W'\in \Ak$ that satisfies
\begin{enumerate}
    \item for all $i\in N$, $|W'\cap P_i| \geqslant |f(E)\cap P_i|$;
    \item for some $i\in N$, $|W'\cap P_i| > |f(E)\cap P_i|$.
\end{enumerate}

\textbullet~ {\bf Condorcet Criterion (CC)}: An ABC rule satisfies Condorcet criterion if for any voting instance $E=(P,k)$ where a Condorcet committee $W$ exists, $f(E)=W$, where the Condorcet committee is a committee $W\in\Ak$ that for all $W'\in\Ak\backslash \{W\}$,
\begin{align*}
    \big|\{ i\in N: |P_i\cap W| > |P_i\cap W'| \}\big| > \frac{n}{2}.
\end{align*}

\section{Two-way Tradeoffs Between DP and Axioms}
\label{sec: 2-tradeoff}

This section investigates the tradeoffs between privacy and voting axioms for ABC rules (2-way tradeoff). To be more specific, we show that all of the axioms defined in Section \ref{sec: prelim} are incompatible with DP. Therefore, we propose approximate versions of them. With the approximate axioms, we establish both upper and lower bounds about their tradeoffs with DP, i.e., the upper and lower bounds of approximate axioms with a given privacy budget $\epsilon\in\mathbb{R}_+$.
All of the missing proofs in this section can be found in Appendix \ref{sec: appendix-A}.


\paragraph{The incompatibility.} First of all, we show the incompatibility between DP and the voting axioms for ABC rules. In fact, for any neutral ABC rule satisfying DP, we have the following lemma.

\begin{lemma}
  \label{lem: dp-neps}
  Given $\epsilon\in\R_+$, let $f\colon \powerset{A}^n\times \N\to \mathcal{R}(\powerset{A})$ be an ABC rule satisfying neutrality and $\epsilon$-DP. Then for any $(P,k)\in\powerset{A}^n\times\N$ and $W_1,W_2\in\Ak$,
  \begin{align*}
    \p[f(P,k)=W_1]\leqslant \erm^{n\epsilon}\cdot \p[f(P,k)=W_2].
  \end{align*}
\end{lemma}

Then, for any voting instance $(P,k)$ and committee $W\in\Ak$, $\p[f(P,k)=W]$ must be strictly greater than $0$, otherwise we will obtain $\p[f(P,k)=W]=0$ for every $W\in \Ak$ by Lemma \ref{lem: dp-neps}, a contradiction. In other words, we have the following corollary.

\begin{corollary}
  \label{cor: dp}
  If a neutral ABC rule $f\colon \powerset{A}^n\times \N\to \mathcal{R}(\powerset{A})$ satisfies $\epsilon$-DP for some $\epsilon\in\R_+$,
  \begin{align*}
    \supp f(P,k)=\Ak,\quad\text{for all }(P,k)\in\powerset{A}^n\times\N.
  \end{align*}
\end{corollary}

Further, consider the voting instance $(P,k)$, where all voters approve the same $k$ alternatives, i.e., there exists a $W\in\Ak$ that $P_j=W$ holds for every $j\in N$. Then it is easy to check that every $W'\neq W$ is Pareto-dominated by $W$. As a result, for any ABC rule $f$ satisfying Pareto efficiency, $\supp f(P,k)=\{W\}$. Then, Corollary \ref{cor: dp} implies that Pareto efficiency cannot be achieved by any ABC rule satisfying DP\footnote{In fact, they are also incompatible in non-neutral cases. See Appendix \ref{subsec: appendix-A1}.}. Similarly, Condorcet criterion and axioms in JR-family are also incompatible with DP (see Appendix \ref{subsec: appendix-A1} for details). Therefore, in the rest of the section, we propose approximate versions of the aforementioned axioms, and discuss their tradeoffs with DP respectively.

\subsection{DP-Proportionality tradeoff}
\label{subsec: 2-JRs}

In this subsection, we discuss the tradeoffs between DP and axioms in JR-family. Intuitively, axioms in JR-family requires the ABC rule to distinguish some of the committees from others. For example, for any ABC rule $f$ satisfying JR, we have $\supp f(P,k)\subseteq \JR(P,k)$ for every voting instance $(P,k)\in\powerset{A}\times\N$, which indicates that a committee $W\in\Ak$ has chance to win the election if and only if $W\in\JR(P,k)$. In other words, any ABC rule $f$ satisfying JR is devoted to distinguish JR committees from non-JR ones. However, Corollary \ref{cor: dp} indicates that JR is unachievable under DP, since for most voting instances $(P,k)$ , $\JR(P,k)\neq \Ak$. Therefore, under DP, we can distinguish JR committees from others only by assigning them different probabilities of winning. Further, the lower bound of the ratio between JR and non-JR committees represents the strength of distinction, i.e., the level of satisfaction with JR. Similar methods also work for PJR and EJR. In the rest of this subsection, the tradeoffs between DP and the axioms in JR-family will be discussed one by one.

\paragraph{DP against JR.}
Firstly, we investigate the tradeoff between DP and JR. To measure how much an ABC rule can approximately approach JR, we propose the following definition, where the parameter $\theta\in\R_+$ represents the level of satisfaction to JR.

\begin{definition}[\bf\boldmath 
  $\theta$-Justified representation, $\theta$-JR for short]
    \label{def: approx-JR}
    Given $\theta\in\mathbb{R}_+$, an ABC rule $f$ satisfies $\theta$-JR if for any voting instance $E=(P,k)$, $W_1\in\JR(P,k)$, and $W_2\notin\JR(P,k)$,
    \begin{align*}
      \p[f(P,k)=W_1] \geqslant \theta\cdot \p[f(P,k)=W_2].
    \end{align*}
\end{definition}

In a word, an ABC rule satisfies $\theta$-JR if the winning probability of any JR committee is at least $\theta$ times the winning probability of any non-JR committee. Therefore, a larger $\theta$ is more desirable and represents a higher level of JR. Notably, $\infty$-JR reduces to its standard form. Further, $\theta$-JR provides a lower bound (which increases with the growth of $\theta$) of the probability that $f(P,k)$ satisfies JR, i.e., 
\begin{align*}
  \p[f(P,k)\in\JR(A,k)] &\geqslant \frac{\theta\cdot |\JR(P,k)|}{\theta\cdot |\JR(P,k)|+\binom{m}{k}-|\JR(P,k)|} \\
  &\geqslant \frac{\theta}{\theta+\binom{m}{k}-1}.
\end{align*}

With this definition, we are ready to discuss the tradeoff between DP and JR. In fact, by Lemma \ref{lem: dp-neps}, we can immediately obtain a trivial upper bound of $\theta$-JR under $\epsilon$-DP, i.e., $\theta\leqslant \erm^{n\epsilon}$. Further, we observe that a pair of neighboring voting instances can have different JR committees, and both of the two committees are unique to their corresponding profiles. Therefore, we have the following theorem, which provides an upper bound much tighter than the trivial upper bound $\erm^{n\epsilon}$.

\begin{theorem}[\bf\boldmath $\theta$-JR, upper bound]
  \label{thm: 2-JR-upper}
  Given any $\epsilon\in\R_+$, there is no ABC rule satisfying $\epsilon$-DP and $\theta$-JR with $\theta>\erm^\epsilon$.
\end{theorem}

\begin{pfsketch}
  Consider a pair of neighboring voting instances $(P,k)$ and $(P',k)$, where $s=\lceil n/k\rceil$:
  \begin{align*}
    P_j &= \begin{cases}
      \{a_1\}, & 1\leqslant j\leqslant s, \\
      \{a_2\}, & s+1\leqslant j\leqslant 2s-1, \\
      A\backslash\{a_1,a_2\}, & \text{otherwise}.
    \end{cases} \\
    P'_j &= \begin{cases}
      \{a_1\}, & 1\leqslant j\leqslant s-1, \\
      \{a_2\}, & s\leqslant j\leqslant 2s-1, \\
      A\backslash\{a_1,a_2\}, & \text{otherwise}.
    \end{cases}
  \end{align*}
  Then $W=\{a_1,a_3,a_4,\ldots,a_{k+1}\}$ satisfies JR only for $(P,k)$, while $W'=\{a_2,a_3,a_4,\ldots,a_{k+1}\}$ satisfies JR only for $(P',k)$. Letting $f$ be an ABC rule satisfying $\epsilon$-DP and $\theta$-JR, we can prove that
  \begin{align*}
    \p[f(P,k)=W] &\geqslant \theta\cdot \erm^{-\epsilon}\cdot \p[f(P',k)=W'] \\
    &\geqslant \theta^2\cdot \erm^{-2\epsilon}\cdot \p[f(P,k)=W],
  \end{align*}
  i.e., $\theta\leqslant \erm^{\epsilon}$, which completes the proof.
\end{pfsketch}

In other words, for any ABC rule satisfying $\epsilon$-DP, its JR level will not exceed $\erm^\epsilon$. Further, we find that $\erm^{\epsilon/2}$ can be achieved by randomized response mechanism , which provides a lower bound of the achievable lower bound of $\theta$-JR under DP. This will be discussed in the last part of the subsection.

\paragraph{DP against PJR.} Secondly, we investigate the tradeoff between DP and PJR. Similar to $\theta$-JR, we design the approximate PJR by introducing a parameter $\rho\in\R_+$, which represents the level of PJR.

\begin{definition}[\bf\boldmath $\rho$-Proportional justified representation, $\rho$- PJR for short]
    \label{def: approx-PJR}
    Given $\rho\in\mathbb{R}_+$, an ABC rule $f$ satisfies $\rho$-PJR if for any voting instance $E=(P,k)$, $W_1\in\PJR(P,k)$, and $W_2\notin\PJR(P,k)$,
    \begin{align*}
      \p[f(P,k)=W_1] \geqslant \rho\cdot \p[f(P,k)=W_2].
    \end{align*}
\end{definition}

In Definition \ref{def: approx-PJR}, the parameter $\rho$ captures the lower bound for the ratio between the winning probability of each pair of PJR and non-PJR committees. Therefore, a larger $\rho$ is more desirable. In particular, when $\rho$ goes to infinity, $\rho$-PJR reduces to its standard form. Accroding to Lemma \ref{lem: dp-neps}, we can derive a trivial upper bound for $\rho$-PJR, i.e., $\rho\leqslant \erm^{n\epsilon}$, and a tighter bound is described in the following theorem.

\begin{theorem}[\bf\boldmath $\rho$-PJR, upper bound]
  \label{thm: 2-PJR-upper}
  Given any $\epsilon\in\R_+$, there is no ABC rule satisfying $\epsilon$-DP and $\rho$-PJR with $\rho>\erm^\epsilon$.
\end{theorem}
\begin{pfsketch}
  Consider the following voting instances $(P,k)$ and $(P',k)$, where $n=s\cdot k$:
  \begin{align*}
    P_j &= \begin{cases}
      \{a_1,a_{k+1}\}, & j=1, \\
      \{a_1\}, & 2\leqslant j\leqslant s, \\
      \{a_2\}, & s+1\leqslant j\leqslant 2s, \\
      ~\cdots \\
      \{a_k\}, & m-k+1\leqslant j\leqslant m.
    \end{cases} \\
    P_j' &= \begin{cases}
      \{a_1,a_{k+2}\}, & j=1, \\
      \{a_1\}, & 2\leqslant j\leqslant s, \\
      \{a_2\}, & s+1\leqslant j\leqslant 2s, \\
      ~\cdots \\
      \{a_k\}, & m-k+1\leqslant j\leqslant m.
    \end{cases}
  \end{align*}
  Let $W=\{a_2,a_3,\ldots,a_k,a_{k+1}\}$ and $W'=\{a_2,a_3,\ldots,a_k,a_{k+2}\}$.
  Then for any $f\colon \themap$ satisfying $\epsilon$-DP and $\rho$-PJR, we can obtain that
  \begin{align*}
    \p[f(P,k)=W] &\geqslant \rho\cdot \erm^{-\epsilon}\cdot \p[f(P',k)=W'] \\
    &\geqslant \rho^2\cdot \erm^{-2\epsilon}\cdot \p[f(P,k)=W],
  \end{align*}
  i.e., $\rho\leqslant \erm^{\epsilon}$, which completes the proof.
\end{pfsketch}

\paragraph{DP against EJR.} Thirdly, we investigate the tradeoff between DP and EJR. Similar to $\theta$-JR and $\rho$-PJR, we have the following definition.

\begin{definition}[\bf\boldmath $\kappa$-Extended justified representation, $\kappa$-EJR for short]
    \label{def: approx-EJR}
    Given $\kappa\in\mathbb{R}_+$, an ABC rule $f$ satisfies $\kappa$-EJR if for any voting instance $E=(P,k)$, $W_1\in\EJR(P,k)$, and $W_2\notin\EJR(P,k)$,
    \begin{align*}
      \p[f(P,k)=W_1] \geqslant \kappa\cdot \p[f(P,k)=W_2].
    \end{align*}
\end{definition}

In Definition \ref{def: approx-EJR}, a larger $\kappa$ is more desirable, since it represents a higher level of EJR. Especially, $\infty$-EJR is equivalent to the standard EJR. Similar to $\theta$-JR and $\rho$-PJR, Lemma \ref{lem: dp-neps} also provides a trivial upper bound for $\kappa$-EJR, i.e., $\kappa\leqslant \erm^{n\epsilon}$. Further, the following theorem provides a tighter bound.

\begin{theorem}[\bf\boldmath $\kappa$-EJR, upper bound]
  \label{thm: 2-EJR-upper}
  Given any $\epsilon\in\R_+$, there is no ABC rule satisfying $\epsilon$-DP and $\kappa$-EJR with $\kappa>\erm^{\epsilon\lceil\frac{n}{k}\rceil}$.
\end{theorem}
\begin{pfsketch}
  Consider the following voting instances $(P,k)$ and $(P',k)$, where $s=\lceil n/k\rceil$:
  \begin{align*}
    P_j &= \begin{cases}
      \{a_1\}, & 1\leqslant j\leqslant s, \\
      \{a_2\}, & s+1\leqslant j\leqslant 2s \\
      ~\cdots \\
      \{a_k\}, & (k-1)s+1\leqslant j\leqslant m.
    \end{cases} \\
    P_j' &= \begin{cases}
      \{a_{k+1}\}, & 1\leqslant j\leqslant s, \\
      \{a_2\}, & s+1\leqslant j\leqslant 2s \\
      ~\cdots \\
      \{a_k\}, & (k-1)s+1\leqslant j\leqslant m.
    \end{cases}
  \end{align*}
  Let $W=\{a_1,a_2,\ldots,a_k\}$ and $W'=\{a_{k+1},a_2,a_3,\ldots,a_k\}$.
  Then we can prove that for any $f\colon \themap$ satisfying $\epsilon$-DP and $\kappa$-EJR,
  \begin{align*}
    \p[f(P,k)=W] &\geqslant \kappa\cdot \erm^{-\lceil\frac{n}{k}\rceil\cdot\epsilon}\cdot \p[f(P',k)=W'] \\
    &\geqslant \kappa^2\cdot \erm^{-2\lceil\frac{n}{k}\rceil\cdot\epsilon}\cdot \p[f(P,k)=W],
  \end{align*}
  i.e., $\kappa\leqslant \erm^{-2\lceil\frac{n}{k}\rceil\cdot\epsilon}$, which completes the proof.
\end{pfsketch}

\paragraph{The lower bounds.} Finally, we discuss the lower bounds for the achievable level of approximate axioms in JR-family under $\epsilon$-DP. In fact, by assigning winning probability to committees dichotomously, the following mechanism achieves a lower bound of approximate JR/PJR/EJR under DP.

\begin{algorithm}[t]
  \caption{Randomized response for JR-family}
  \label{algo: RR-JR}
  \KwIn{Profile $P$, Committee Size $k$, Noise Level $\epsilon$, Axiom $\mathcal{A}$}
  \KwOut{Winning Committee $W_{win}$}
  \For{$W\in \Ak$}{
    \If{$W$ satisfies $\mathcal{A}$}{
      $\chi(W)\leftarrow 1$\;
    }\Else{
      $\chi(W)\leftarrow 0$\;
    }
  }
  Compute the probability distribution $p$, such that $p(W)\propto \erm^{\chi(W)\cdot\epsilon/2}$ for all $W\in\Ak$\;
  Sample $W_{win}\sim p$\;
  \Return{$W_{win}$}
\end{algorithm}

Intuitively, Mechanism \ref{algo: RR-JR} equally increases the winning probability of every committee satisfying JR/PJR/EJR. Formally, we have the following proposition.

\begin{proposition}
  \label{prop: 2-JR-lower}
  For any noise level $\epsilon\in\R_+$, voting axiom $\mathcal{A}\in\{\JR, \PJR, \EJR\}$ and voting instance $(P,k)\in\powerset{A}^n\times\N$, Mechanism \ref{algo: RR-JR} satisfies $\epsilon$-DP and $\erm^{\epsilon/2}$-$\mathcal{A}$.
\end{proposition}

\subsection{DP-Efficiency tradeoff}

\label{subsec: 2-eff}

In this subsection, we discuss the tradeoffs between DP and efficiency axioms, includeing Pareto efficiency (PE) and Condorcet criterion (CC). For their approximate versions, we extend the probabilistic Pareto efficiency and probabilistic Condorcet criterion in \citep{li2022differentially}.

\paragraph{DP against PE.} First, we discuss the tradeoff between DP and PE. As is mentioned previously, PE requires the ABC rule to avoid selecting a Pareto-dominated committee as the winner for any input profile, which cannot be achieved under DP. Therefore, we construct the following definition, where we introduce a parameter $\beta$ to quantify the level of Pareto efficiency.

\begin{definition}[\bf\boldmath $\beta$-Pareto efficiency, $\beta$-PE for short]
    \label{def: beta-pareto}
    Given $\beta\in\R_+$, an ABC rule $f$ satisfies Pareto efficiency if for all voting instance $E=(P,k)$, where exist two distinct committee $W_1, W_2\in\powerset{A}$ that $W_1$ Pareto dominates $W_2$, 
    \begin{align*}
        \p[f(E)=W_1] \geqslant \beta\cdot \p[f(E)=W_2].
    \end{align*}
\end{definition}

In other words, an ABC rule satisfies $\beta$-PE if for each pair of Pareto-dominant and dominated committees, the ratio of their winning probability can be lower bounded by $\beta$. Therefore, a larger $\beta$ is more desirable, since it represents a higher level of Pareto efficiency. It is also easy to check that $\beta$-PE reduces to the standard Pareto efficiency when $\beta$ goes to infinity.

Now, we can discuss the tradeoff between DP and Pareto efficiency in the sense of Definition \ref{def: beta-pareto}. Formally, the following theorem provides an upper bound of this tradeoff.

\begin{theorem}[\bf\boldmath $\beta$-PE, upper bound]
    \label{thm: 2-PE-upper}
    Given any $\epsilon\in\R_+$, there is no neutral ABC rule satisfying $\epsilon$-DP and $\beta$-PE with $\beta>\erm^{\frac{\epsilon}{k}}$.
\end{theorem}
\begin{pfsketch}
  Consider the following voting instance $E=(P,k)$, where $m\geqslant n+2k-1$ and $\{a_1,\ldots,a_k,b_1,\ldots,b_{n-1}, c_1,\ldots,c_k\}\subseteq A$.
    \begin{align*}
      P_j = \begin{cases}
        \{a_1,a_2,\ldots,a_k\}, & j=1, \\
        \{a_1,a_2,\ldots,a_k,b_1\}, & j=2, \\
        ~\cdots \\
        \{a_1,a_2,\ldots,a_k,b_1,b_2,\ldots,b_{n-1}\}, & j=n. \\
      \end{cases}
    \end{align*}
    Then we construct a series of distinct committees $W_{p,q}$, where $1\leqslant p\leqslant k$ and $1\leqslant q\leqslant n$.
    \begin{align*}
      W_{p,q} = \begin{cases}
        \{a_1,\ldots,a_{k-p+1},c_1,\ldots,c_{p-1}\}, & q=1, \\
        \{a_1,\ldots,a_{k-p},b_{q-1},c_1,\ldots,c_{p-1}\}, & q>1.
      \end{cases}
    \end{align*}
    Further, we can prove that $\{W_{p,q}\}_{p\leqslant k, q\leqslant n}$ forms a chain under the Pareto-dominance relationship. Consider $W_{1,1}=\{a_1,a_2,\ldots,a_k\}$ and $W_{k+1,1}=\{c_1,c_2,\ldots,c_k\}$, we have
    \begin{align*}
      \p[f(P,k)=W_{1,1}] \geqslant \beta^{nk}\cdot \p[f(P,k)=W_{k+1,1}],
    \end{align*}
    for any neutral ABC rule $f\colon\themap$ satisfying $\beta$-PE. Further, if $f$ satisfies $\epsilon$-DP, we have
    \begin{align*}
      \p[f(P,k)=W_{1,1}] \leqslant \erm^{n\epsilon}\cdot \p[f(P,k)=W_{k+1,1}].
    \end{align*}
    In other words, we have $\beta\leqslant \erm^{\frac{\epsilon}{k}}$, which completes the proof.
\end{pfsketch}

Theorem \ref{thm: 2-PE-upper} shows an upper bound of the tradeoff between DP and Pareto efficiency. Next, we propose a mechanism to provide a lower bound of the achievable level of $\beta$-PE under $\epsilon$-DP, shown in Mechanism \ref{algo: AVExp}. Technically, Mechanism \ref{algo: AVExp} is equivalent to the exponential mechanism using the approval voting (AV) score as utility metric, where AV score is defined in the following equation.
\begin{align*}
    {\rm AV}_P(W) = \sum\limits_{c\in W} |\{j\in N:c\in P_j\}| = \sum\limits_{j\in N}|P_j\cap W|.
\end{align*}

\begin{algorithm}[ht]
    \caption{Exponential mechanism with AV score}
    \label{algo: AVExp}
    \KwIn{Profile $P$, Committee Size $k$, Noise Level $\epsilon$}
    \KwOut{Winning Committee $W_{win}$}
    \For{$a\in A$}{
      $\chi(a)\leftarrow |\{j\in N: a\in P_j\}|$\;
    }
    $W_{win}\leftarrow \varnothing$\;
    \While{$|W_{win}|<k$}{
      Compute the probability distribution $p$, such that $p(a)\propto \erm^{\chi(a)\cdot\epsilon/(2k)}$ for all $a\in A\backslash W_{win}$\;
      Sample $a\sim p$\;
      $W_{win}\leftarrow W_{win}\cup\{a\}$\;
    }
    \Return{$W_{win}$}
\end{algorithm}

It is not hard to check that the time complexity of Mechanism \ref{algo: AVExp} is $O(m(n+k))$. Further, the following proposition provides the DP and PE level of the mechanism.

\begin{proposition}[\bf\boldmath $\beta$-PE, lower bound]
    \label{prop: 2-PE-lower}
    Given any $\epsilon\in\mathbb{R}_+$, Mechanism \ref{algo: AVExp} satisfies $\epsilon$-DP and $\erm^{\frac{\epsilon}{2k}}$-PE.
\end{proposition}

\paragraph{DP against CC.} Secondly, we study the tradeoff between DP and CC. Similar to $\beta$-PE, we propose an approximate version of CC by introducing a parameter $\eta\in\R_+$ to quantify the level of CC.

\begin{definition}[\bf\boldmath $\eta$-Condorcet criterion, $\eta$-CC for short]
    \label{def: approx-CC}
    An ABC rule satisfies $\eta$-Condorcet criterion if for any voting instance $E=(P,k)$ where a Condorcet committee $W_c$ exists,
    \begin{align*}
        \p[f(E)=W_c] \geqslant \eta\cdot \p[f(E)=W],\quad\text{for all}~W\in\Ak\backslash \{W_c\}.
    \end{align*}
\end{definition}

It is worth noting that in Definition \ref{def: approx-CC}, a larger $\eta$ is more desirable. When $\eta$ goes to infinity, $\eta$-CC approaches the standard form of Condorcet criterion.

With the definition above, we can quantify the tradeoff between DP and Condorcet criterion by investigating the relationship between $\eta$ and the privacy budget $\epsilon$. In fact, the following theorem illustrates an upper bound of $\eta$ under $\epsilon$-DP.

\begin{theorem}[\bf\boldmath $\eta$-CC, upper bound]
    \label{thm: 2-CC-upper}
    Given any $\epsilon\in\R_+$, there is no ABC rule satisfying $\epsilon$-DP and $\eta$-CC with $\eta>\erm^{\epsilon}$. 
\end{theorem}
\begin{pfsketch}
    Consdier the voting instances $(P,k)$ and $(P',k)$, where $n=2t+1$ and there exists a $W\in\powerset{A,k-1}$ that
    \begin{align*}
      P_j &= \begin{cases}
        W\cup \{a_1\}, & 1\leqslant j\leqslant t+1, \\
        W\cup \{a_2\}, & t+2\leqslant j\leqslant 2t+1.
      \end{cases} \\
      P_j' &= \begin{cases}
        W\cup \{a_1\}, & 1\leqslant j\leqslant t, \\
        W\cup \{a_2\}, & t+1\leqslant j\leqslant 2t+1.
      \end{cases}
    \end{align*}
    Then $W\cup \{a_1\}$ and $W\cup \{a_2\}$ are Condorcet committees for $P$ and $P'$, respectively. Therefore, for any ABC rule $f\colon\themap$ satisfying $\epsilon$-DP and $\eta$-CC, we have
    \begin{align*}
      \p[f(P,k)=W\cup \{a_1\}] \geqslant&\; \erm^{-\epsilon}\cdot \eta\cdot \p[f(P',k)=W\cup \{a_2\}] \\
      \geqslant&\; \erm^{-2\epsilon}\cdot \eta^2\cdot \p[f(P,k)=W\cup \{a_1\}],
    \end{align*}
    i.e., $\eta\leqslant \erm^{\epsilon}$, which completes the proof.
\end{pfsketch}

Theorem \ref{thm: 2-CC-upper} provides an upper bound of the achievable level of CC under $\epsilon$-DP. Further, we show that the bound can be achieved by Mechanism \ref{algo: Condorcet-RR} (randomized response for Condorcet committee). 

\begin{algorithm}[ht]
    \caption{Randomized response for Condorcet committee}
    \label{algo: Condorcet-RR}
    \KwIn{Profile $P$, Committee Size $k$, Noise Level $\epsilon$}
    \KwOut{Winning Committee $W_{win}$}
    \If{Condorcet committee $W_c$ exists}{
        Compute the distribution $p$ on $\Ak$, such that $p(W)=\begin{cases}
        \frac{\erm^{\epsilon}}{\erm^{\epsilon}+\binom{m}{k}-1}, & W = W_c \\
        \frac{1}{\erm^{\epsilon}+\binom{m}{k}-1}, & \text{\rm otherwise}
    \end{cases}$\;
    }\Else{
        Let $p$ be the uniform distribution on $\Ak$, i.e., $p(W)=\binom{m}{k}^{-1}$, for all $W\in \Ak$\;
    }
    Sample $W_{win}\sim p$\;
    \Return{$W_{win}$}
\end{algorithm}

Technically, we design this mechanism by applying randomized response mechanism to the Condorcet committee. Formally, the following Proposition shows that the mechanism satisfies both DP and approximate CC, which provides a lower bound for achievable level of $\eta$-CC under $\epsilon$-DP.

\begin{proposition}[\bf\boldmath $\eta$-CC, lower bound]
    \label{prop: 2-CC-lower}
    Given any $\epsilon\in\mathbb{R}_+$, Mechanism \ref{algo: Condorcet-RR} satisfies $\epsilon$-DP and $\erm^{\epsilon}$-CC.
\end{proposition}

In fact, Proposition \ref{prop: 2-CC-lower} indicates that the upper bound illustrated in Theorem \ref{thm: 2-CC-upper} is tight.

\section{Three-way Tradeoffs Among DP and Axioms}
\label{sec: 3-tradeoff}

This section investigates the 3-way tradeoffs among DP and voting axioms for ABC rules. Precisely, we examine the distinction between the axiom tradeoffs in classical social choice theory and the axiom tradeoffs under DP. We take all the axioms mentioned in the previous sections into consideration. It is worth noting that each pairwise combination of axioms (the standard versions) is compatible, except for pairs consisting of CC and another axiom in JR-family. In other words, the incompatibility between CC and JR-family is not introduce by DP. In fact, letting an ABC rule $f$ satisfies $\epsilon$-DP, $\eta$-CC, and $\theta$-JR, for a voting instance $(P,k)$ that the Condorcet committee $W_c$ does not satisfy JR, we have
\begin{align*}
  \p[f(P,k)=W_c] &\geqslant \eta\cdot \p[f(P,k)=W_0] \\
  &\geqslant \eta\cdot\theta\cdot \p[f(P,k)=W_c],
\end{align*}
where $W_0\in\JR(P,k)$. In other words, we have an upper bound which is independent with $\epsilon$, i.e., $\eta\cdot\theta\leqslant 1$. Further, this upper bound is tight, since it can be achieved by Mechanism \ref{algo: Condorcet-RR}. Similar statement are also true for PJR and EJR. Therefore, we only study the 3-way tradeoffs among DP and other pairs of axioms in this section. All of the missing proofs in this section and more detailed discussions on the compatibility between axioms without DP can be found in Appendix \ref{sec: appendix-B}.

Firstly, we investigate the tradeoffs among axioms in JR-family. In fact, without DP, it is quite evident that JR, PJR, and EJR are compatible with each other. But when DP is required, the purposes of $\theta$-JR, $\rho$-PJR, and $\kappa$-EJR start to diverge. In fact, consider the following voting instance $(P,k)$ presented in Example 4.6 of~\cite{lackner2023multi}, where the set of voters can be divided into $k$ disjoint groups $V_1,V_2,\ldots,V_k$ satisfying $|V_t|=n/k$ for all $t$ and
\begin{align*}
  P_j = \begin{cases}
    \{a_1,a_{k+1},a_{k+2},\ldots,a_{2k}\}, & j\in V_1, \\
    \{a_2,a_{k+1},a_{k+2},\ldots,a_{2k}\}, & j\in V_2, \\
    ~\cdots \\
    \{a_k,a_{k+1},a_{k+2},\ldots,a_{2k}\}, & j\in V_k.
  \end{cases}
\end{align*}

\begin{figure}[h]
  \centering
  \begin{tikzpicture}[yscale=1.0,xscale=1.7]
    \filldraw[fill=cyan!10!white, draw=black] (0,0.5) rectangle (4.0,1.0);
    \node at (2.0, 0.75) {$a_{k+1}$};
    \filldraw[fill=cyan!10!white, draw=black] (0,1.0) rectangle (4.0,1.5);
    \node at (2.0, 1.25) {$a_{k+2}$};
    \filldraw[fill=cyan!10!white, draw=black] (0,1.5) rectangle (4.0,2.0);
    \node at (2.0, 1.75) {$\cdots$};
    \filldraw[fill=cyan!10!white, draw=black] (0,2.0) rectangle (4.0,2.5);
    \node at (2.0, 2.25) {$a_{2k}$};

    \filldraw[fill=orange!10!white, draw=black] (0,0) rectangle (0.8,0.5);
    \node at (0.4, 0.25) {$a_1$};
    \node at (0.4, -0.35) {$V_1$};
    \filldraw[fill=orange!10!white, draw=black] (0.8, 0) rectangle (1.6,0.5);
    \node at (1.2, 0.25) {$a_2$};
    \node at (1.2, -0.35) {$V_2$};
    \filldraw[fill=orange!10!white, draw=black] (1.6, 0) rectangle (2.4,0.5);
    \node at (2.0, 0.25) {$a_3$};
    \node at (2.0, -0.35) {$V_3$};
    \filldraw[fill=orange!10!white, draw=black] (2.4, 0) rectangle (3.2,0.5);
    \node at (2.8, 0.25) {$\cdots$};
    \node at (2.8, -0.35) {$\cdots$};
    \filldraw[fill=orange!10!white, draw=black] (3.2, 0) rectangle (4.0,0.5);
    \node at (3.6, 0.25) {$a_{k}$};
    \node at (3.6, -0.35) {$V_{k}$};
  \end{tikzpicture}
  \caption{\boldmath Diagram corresponding to the profile $P$.}
  \label{fig: JRExample}
\end{figure}

The profile $P$ can be visualized as Figure \ref{fig: JRExample}, where a block labeled with alternative $a$ is placed above a voter group $V$ indicates that $V$ approves $a$ in $P$. It is easy to check that $W_1=\{a_{k+1},a_{k+2},\ldots,a_{2k}\}$ is the unique committee satisfying EJR, and $W_2=\{a_1,a_2,\ldots,a_k\}$ satisfies PJR. Intuitively, for such a voting instance, any ABC rule $f$ that achieves the optimal $\rho$-PJR under $\epsilon$-DP should be devoted to maximize the minimum winning probability of $W_1, W_2$, and miminize the maximum winning probability of the other committees. However, $\kappa$-EJR only focuses on maximizing the winning probabilities of $W_1$. In other words, $\kappa$-EJR requires the ABC rule to distinguish $W_1$ from $W_2$, while $\rho$-PJR requires the ABC rule to treat them as the same, which indicates that DP introduces additional tradeoffs between PJR and EJR. Formally, the following two theorems quantify this kind of tradeoffs among JR-family.

\begin{theorem}[\bf\boldmath JR against PJR/EJR]
  \label{thm: 3-JR-PJR-upper}
  Given any $\epsilon\in\R_+$, there is no ABC rule satisfying
  \begin{enumerate}[{\rm (1)}]
    \item $\epsilon$-DP, $\theta$-JR, and $\rho$-PJR with $\theta\cdot \rho > \erm^{\epsilon}$;
    \item $\epsilon$-DP, $\theta$-JR, and $\kappa$-EJR with $\theta\cdot \kappa > \erm^{\epsilon}$.
  \end{enumerate}
\end{theorem}
\begin{pfsketch}
    For both (1) and (2), consider the following neighboring voting instances $(P,k)$ and $(P',k)$, where $s=\lceil 2n/k\rceil$:
    \begin{align*}
    P_j &= \begin{cases}
      \{a_1,a_2\}, & 1\leqslant j\leqslant s, \\
      \{a_1,a_3\}, & s+1\leqslant j\leqslant 2s-1, \\
      \{a_4,a_5\}, & 2s\leqslant j\leqslant 3s-1, \\
      \{a_4\}, & \text{otherwise}.
    \end{cases} \\
    P_j' &= \begin{cases}
      \{a_1,a_2\}, & 1\leqslant j\leqslant s-1, \\
      \{a_1,a_3\}, & s\leqslant j\leqslant 2s-1, \\
      \{a_4,a_5\}, & 2s\leqslant j\leqslant 3s-1, \\
      \{a_4\}, & \text{otherwise}.
    \end{cases}
    \end{align*}
    Then we can construct the following committees.
    \begin{align*}
        W_0 &=\{a_1,a_4,a_6,a_7,\ldots,a_{k+3}\}, \\
        W_1 &=\{a_1,a_2,a_4,a_5,a_6,\ldots,a_{k+1}\}, \\
        W_1' &=\{a_1,a_3,a_4,a_5,a_6,\ldots,a_{k+1}\}.
    \end{align*}
    With the constructions above, we can finally prove that for any $f\colon\themap$ satisfying $\epsilon$-DP, $\theta$-JR, and $\rho$-PJR,
    \begin{align*}
        \p[f(P,k)=W_1] &\geqslant \rho\cdot \p[f(P,k)=W_0] \\
        &\geqslant \rho\cdot\theta\cdot\erm^{-\epsilon}\cdot \p[f(P',k)=W_1'] \\
        &\geqslant \rho^2\cdot\theta^2\cdot\erm^{-2\epsilon}\cdot \p[f(P,k)=W_1],
    \end{align*}
    i.e., $\rho\cdot\theta\leqslant\erm^\epsilon$, which completes the proof.
\end{pfsketch}

\begin{theorem}[\bf\boldmath PJR against EJR]
  \label{thm: 3-PJR-EJR-upper}
  Given any $\epsilon\in\R_+$, there is no ABC rule satisfying $\epsilon$-DP, $\rho$-PJR, and $\kappa$-EJR with $\rho\cdot\kappa > \erm^{\epsilon\lceil \frac{n}{k}\rceil}$.
\end{theorem}
\begin{pfsketch}
    Consider the following voting instances $(P,k)$ and $(P'k)$, where $s=\lceil n/k\rceil$:
    \begin{align*}
        P_j &= \begin{cases}
            \{a_1,a_2,\ldots,a_k,a_{k+1}\}, & 1\leqslant j\leqslant s, \\
            \{a_1,a_2,\ldots,a_k,a_{k+2}\}, & s+1\leqslant j\leqslant 2s, \\
            ~\cdots \\
            \{a_1,a_2,\ldots,a_k,a_{2k}\}, & (k-1)s+1\leqslant j\leqslant n.
        \end{cases} \\
        P_j' &= \begin{cases}
            \{a_{k+1},a_{2k+1},a_{2k+2},\ldots,a_{3k-1}\}, & 1\leqslant j\leqslant s, \\
            P_j, & \text{otherwise}.
        \end{cases}
    \end{align*}
    Then we can construct the following committees.
    \begin{align*}
        W_0 &=\{a_{k+1},a_{k+2},\ldots,a_{2k}\}, \\
        W_1 &=\{a_1,a_2,\ldots,a_k\}, \\
        W_1' &=\{a_1,a_2,\ldots,a_{k-1},a_{2k+1}\}.
    \end{align*}
    Finally, for any $f\colon\themap$ satisfying $\epsilon$-DP, $\rho$-PJR, and $\kappa$-EJR, we can prove that
    \begin{align*}
        \p[f(P,k)=W_1] &\geqslant \kappa\cdot \p[f(P,k)=W_0] \\
        &\geqslant \kappa\cdot\rho\cdot\erm^{-s\epsilon}\cdot \p[f(P',k)=W_1'] \\
        &\geqslant \kappa^2\cdot\rho^2\cdot\erm^{-2s\epsilon}\cdot \p[f(P,k)=W_1],
    \end{align*}
    i.e., $\kappa\cdot \rho\leqslant \erm^{\lceil n/k \rceil\cdot\epsilon}$, which completes the proof.
\end{pfsketch}

Secondly, we investigate the tradeoffs between PE and axioms in JR-family. Under the deterministic settings (without DP), the proportional approval voting (PAV) rule satisfies both PE and EJR \citep{aziz2017justified}, which indicates that the standard PE are compatible with all axioms in JR-family. However, when DP is required, the tradeoffs between $\beta$-PE and approximate axioms in JR-family begin to appear. Intuitively, when a Pareto-dominated committee satisfies JR (or PJR/EJR), $\beta$-PE requires the ABC rule to distinguish it from its Pareto-dominator, while $\theta$-JR (or $\rho$-PJR/$\kappa$-EJR) treats them as the same. Formally, the following theorem provides upper bounds for the tradeoffs between PE and JR-family under $\epsilon$-DP.

\begin{theorem}[\bf PE against JR-family]
  \label{thm: 3-PE-JR-upper}
  Given any $\epsilon\in\R_+$, there is no neutral ABC rule satisfying
  \begin{enumerate}[{\rm (1)}]
    \item $\epsilon$-DP, $\beta$-PE, and $\theta$-JR with $\beta^{nk-1}\cdot\theta> \erm^{n\epsilon}$;
    \item $\epsilon$-DP, $\beta$-PE, and $\rho$-PJR with $\beta^{nk-1}\cdot\rho> \erm^{n\epsilon}$;
    \item $\epsilon$-DP, $\beta$-PE, and $\kappa$-EJR with $\beta^{nk-1}\cdot\kappa> \erm^{n\epsilon}$.
  \end{enumerate}
\end{theorem}
\begin{pfsketch}
    Here we only take (1) as example. Consider the following voting instance $E=(P,k)$, where $m\geqslant n+2k-1$ and $\{a_1,\ldots,a_k,b_1,\ldots,b_{n-1},c_1,\ldots,c_k\}\subseteq A$.
    \begin{align*}
    P_j = \begin{cases}
            \{a_1,a_2,\ldots,a_k\}, & j=1, \\
            \{a_1,a_2,\ldots,a_k,b_1\}, & j=2, \\
            ~\cdots \\
            \{a_1,a_2,\ldots,a_k,b_1,b_2,\ldots,b_{n-1}\}, & j=n. \\
        \end{cases}
    \end{align*}
    By constructing the same series of committees $W_{p,q}$ as the proof of Theorem \ref{thm: 2-PE-upper}, we find that $W_{1,1}=\{a_1,a_2,\ldots,a_k\}$ satisfies JR for $(P,k)$, while $W_{k+1,1}=\{c_1,c_2,\ldots,c_k\}$ does not satisfies JR for $(P,k)$.
    Therefore, for any neutral ABC rule $f\colon\themap$ satisfying $\epsilon$-DP, $\beta$-PE, and $\theta$-JR, we can prove that
    \begin{align*}
        \p[f(P,k)=W_{1,1}] &\leqslant \erm^{n\epsilon}\cdot \p[f(P,k)=W_{k+1,1}], \text{and} \\
        \p[f(P,k)=W_{1,1}] &\geqslant \beta^{kn-1}\cdot\theta\cdot \p[f(P,k)=W_{k+1,1}].
    \end{align*}
    In other words, we have $\beta^{nk-1}\cdot\theta\leqslant \erm^{n\epsilon}$, which completes the proof. Statements (2) and (3) can be proved similarly.
\end{pfsketch}

Finally, we discuss the tradeoff between PE and CC under DP. Similar to the pairwise combinations of axioms above, the standard form of PE and CC are compatible with each other, since the Condorcet committee is never Pareto-dominated by any other committee. However, under the constraint of DP, $\eta$-CC ignores the Pareto-dominance relationships among committees other than the Condorcet committee. Further, we have the following theorem.

\begin{theorem}[\bf PE against CC]
  \label{thm: 3-PE-CC-upper}
  Given any $\epsilon\in\R_+$, there is no neutral ABC rule satisfying $\epsilon$-DP, $\beta$-PE, and $\eta$-CC with $\beta^{nk-1}\cdot\eta\leqslant \erm^{n\epsilon}$.
\end{theorem}
\begin{pfsketch}
    Consider the same $(P,k)$ and $W_{p,q}$ in the proof of Theorem \ref{thm: 2-PE-upper}. It is easy to check that $W_{1,1}$ is the Condorcet committee for $(P,k)$. Therefore, for any neutral ABC rule $f\colon\themap$ satisfying $\epsilon$-DP, $\beta$-PE, and $\eta$-CC, we have
    \begin{align*}
        \p[f(P,k)=W_{1,1}] &\leqslant \erm^{n\epsilon}\cdot \p[f(P,k)=W_{k+1,1}], \text{and} \\
        \p[f(P,k)=W_{1,1}] &\geqslant \eta\cdot\beta^{nk-1}\cdot \p[f(P,k)=W_{k+1,1}].
    \end{align*}
    In other words, we have $\eta\cdot\beta^{nk-1}\leqslant \erm^{n\epsilon}$.
\end{pfsketch}

\section{Conclusion and Future Work}
\label{sec: future}

In the paper, we investigated the tradeoffs among DP and several axioms in approval-based committee voting, including JR, PJR, EJR, PE, and CC. We found that DP is significantly incompatible with all axioms, and quantified their tradeoffs against DP (2-way tradeoff). Further, we studied the tradeoffs between axioms under DP (3-way tradeoff). By capturing upper bounds for 3-way tradeoffs, we showed that DP actually introduces additional tradeoffs among axioms. However, some of the bounds in the paper are not tight. Thus, deriving tighter bounds for both 2-way and 3-way tradeoffs can be a promising future direction. Besides, the time complexities of Mechanism \ref{algo: RR-JR} and \ref{algo: Condorcet-RR} are exponential. Designing mechanisms with polynomial time complexity is also interesting for future work.

\bibliography{ref}

\appendix

\newenvironment{proofwithname}[1][Proof]{{\noindent\it \bf #1.}\ }{\hfill $\square$\par}

\section{Missing Proofs in Section \ref{sec: 2-tradeoff}}

\label{sec: appendix-A}

\subsection{Detailed discussions on incompatibility between DP and axioms}
\label{subsec: appendix-A1}

First, we prove Lemma \ref{lem: dp-neps}.

\begin{proofwithname}[Proof of Lemma \ref{lem: dp-neps}]
  Let $W_1=\{a_1,a_2,\ldots,a_k\}$, $W_2=\{b_1,\ldots,b_k\}$. Then we can construct a permutation $\sigma$ on $A$, such that $\sigma(W_1)=W_2$, $\sigma(W_2)=W_1$, and satisfies $\sigma(a)=a$ for every $a\in A\backslash \{W_1\cup W_2\}$. By neutrality, we have
  \begin{align*}
    \p[f(P,k)=W_2] &= \p[f(\sigma\cdot P, k)=W_1].
  \end{align*}
  Since $W_1\neq W_2$, there are $n$ different votes in $P$ and $\sigma\cdot P$. Therefore, $\epsilon$-DP indicates that
  \begin{align*}
    \p[f(P,k)=W_1] &\leqslant \erm^{n\epsilon}\cdot \p[f(\sigma\cdot P,k)=W_1] \\
    &= \erm^{n\epsilon}\cdot \p[f(P,k)=W_2],
  \end{align*}
  which completes the proof.
\end{proofwithname}

Further, the following lemma can be regarded as a more general version of Corollary \ref{cor: dp}, where neutrality is no longer required.

\begin{lemma}
  \label{lem: supp-dp}
  If an ABC rule $f\colon\themap$ satisfies $\epsilon$-DP for some $\epsilon\in\R_+$, for any voting instances $(P,k)$ and $(P',k)$,
  \begin{align*}
    \supp f(P,k)=\supp f(P',k).
  \end{align*}
\end{lemma}

\begin{proof}
  Let $(P,k)$ and $(P',k)$ be two voting instances. Since $f$ satisfies $\epsilon$-DP, for any committee $W\in\supp f(P,k)$, we have
  \begin{align*}
    \p[f(P',k)=W] &\geqslant \erm^{-|j\in N: P_j\neq P_j'|\cdot\epsilon}\cdot \p[f(P,k)=W] \\
    &\geqslant \erm^{-n\epsilon}\cdot \p[f(P,k)=W] \\
    &>0.
  \end{align*}
  In other words, we have $W\in\supp f(P',k)$, which indicates that $\supp f(P,k)\subseteq \supp f(P',k)$. Similarly, we have $\supp f(P',k)\subseteq \supp f(P,k)$, which completes the proof.
\end{proof}

In fact, without the assumption of neutrality, we can still obtain the incompatibility between DP and axioms, as shown in the following corollary.

\begin{corollary}
  If an ABC rule $f\colon\themap$ satisfies $\epsilon$-DP for some $\epsilon\in\R_+$, then $f$ does not satisfy JR, PE, or CC.
\end{corollary}
\begin{proof}
  We only need to prove that any ABC rule satisfying JR, PE, or CC cannot satisfy DP.

  Firstly, let $f$ be an ABC rule satisfying JR. Then for any $W,W'\in\Ak$, we can construct two distinct voting instances $(P,k)$ and $(P',k)$, such that $W$ and $W'$ are their unique JR committee, respectively (each consists of $k$ disjoint $1$-cohesive groups). Then $\supp f(P,k) = \{W\}\neq \{W'\}\supp f(P',k)$, which indicates that $f$ does not satisfy $\epsilon$-DP for any $\epsilon$.

  Secondly, let $f$ be an ABC rule satisfying PE. Then for any committee $W\in\Ak$, considering the voting instance $(P,k)$, where $P_j=A\backslash W$, we have $\p[f(P,k)=W]=0$, since $W$ must be Pareto dominated by another committee. In other words, we have $\p[f(P,k)=W]=0$ for all $W\in\Ak$, a contradiction.

  Finally, for CC, it is easy to construct two voting instance $(P,k)$ and $(P',k)$ with different Condorcet committees $W$ and $W'$. Then $\supp f(P,k)\neq \supp f(P,k)$. By Lemma \ref{lem: supp-dp}, $f$ does not satisfy $\epsilon$-DP for any $\epsilon$, which completes the proof.
\end{proof}

\subsection{Missing proofs in subsection \ref{subsec: 2-JRs}}

Lemma \ref{lem: dp-neps} has been proved in the previous subsection (Appendix A.1).

\begin{proofwithname}[Proof of Theorem \ref{thm: 2-JR-upper}]
  Suppose $f\colon \themap$ be an ABC rule satisfying $\epsilon$-DP and $\theta$-JR.
  Consider the following voting instances $(P,k)$ and $(P',k)$, where $s=\lceil n/k\rceil$:
  \begin{align*}
    P_j &= \begin{cases}
      \{a_1\}, & 1\leqslant j\leqslant s, \\
      \{a_2\}, & s+1\leqslant j\leqslant 2s-1, \\
      A\backslash\{a_1,a_2\}, & \text{otherwise}.
    \end{cases} \\
    P'_j &= \begin{cases}
      \{a_1\}, & 1\leqslant j\leqslant s-1, \\
      \{a_2\}, & s\leqslant j\leqslant 2s-1, \\
      A\backslash\{a_1,a_2\}, & \text{otherwise}.
    \end{cases}
  \end{align*}
  It is not hard to see that $P_s\neq P_s'$ and $P_{-s}=P_{-s}'$, i.e., $P$ and $P'$ are neighboring profiles. Further, voters $\{1,2,\ldots,s\}$ and $\{s,s+1,\ldots,2s-1\}$ form $1$-cohesive groups for $(P,k)$ and $(P',k)$, respectively. Therefore, committee $W=\{a_1,a_3,a_4,\ldots,a_{k+1}\}$ satisfies JR for $(P,k)$, but does not satisfy JR for $(P',k)$. Similarly, $W'=\{a_2,a_3,a_4,\ldots,a_{k+1}\}$ satisfies JR for $(P',k)$, but does not satisfy JR for $(P,k)$. Then we have
  \begin{align*}
    \p[f(P,k)=W] &\geqslant \theta\cdot \p[f(P,k)=W'] \tag*{($\theta$-JR)} \\
    &\geqslant \theta\cdot \erm^{-\epsilon}\cdot \p[f(P',k)=W'] \tag*{($\epsilon$-DP)} \\
    &\geqslant \theta^2\cdot \erm^{-\epsilon}\cdot \p[f(P',k)=W] \tag*{($\theta$-JR)} \\
    &\geqslant \theta^2\cdot \erm^{-2\epsilon}\cdot \p[f(P,k)=W]. \tag*{($\epsilon$-DP)}
  \end{align*}
  In other words, $\theta^2\cdot \erm^{-2\epsilon}\leqslant 1$, i.e., $\theta\leqslant \erm^\epsilon$, which completes the proof.
\end{proofwithname}

\begin{proofwithname}[Proof of Theorem \ref{thm: 2-PJR-upper}]
  Suppose $f\colon \themap$ be an ABC rule satisfying $\epsilon$-DP and $\rho$-PJR.
  Consider the following voting instances $(P,k)$ and $(P',k)$, where $n=s\cdot k$:
  \begin{align*}
    P_j &= \begin{cases}
      \{a_1,a_{k+1}\}, & j=1, \\
      \{a_1\}, & 2\leqslant j\leqslant s, \\
      \{a_2\}, & s+1\leqslant j\leqslant 2s, \\
      ~\cdots \\
      \{a_k\}, & m-k+1\leqslant j\leqslant m.
    \end{cases} \\
    P_j' &= \begin{cases}
      \{a_1,a_{k+2}\}, & j=1, \\
      \{a_1\}, & 2\leqslant j\leqslant s, \\
      \{a_2\}, & s+1\leqslant j\leqslant 2s, \\
      ~\cdots \\
      \{a_k\}, & m-k+1\leqslant j\leqslant m.
    \end{cases}
  \end{align*}
  It is quite evident that $P$ and $P'$ are neighboring profiles, since they only differ on the first voter's vote. Further, we claim that $W=\{a_2,a_3,\ldots,a_k,a_{k+1}\}$ and $W'=\{a_2,a_3,\ldots,a_k,a_{k+2}\}$ satisfies PJR for $P$ and $P'$, respectively. Besides, it is not hard to check that $W\notin \PJR(P',k)$ and $W'\notin \PJR(P,k)$. Therefore,
  \begin{align*}
    \p[f(P,k)=W] &\geqslant \rho\cdot \p[f(P,k)=W'] \tag*{($\rho$-PJR)} \\
    &\geqslant \rho\cdot \erm^{-\epsilon}\cdot \p[f(P',k)=W'] \tag*{($\epsilon$-DP)} \\
    &\geqslant \rho^2\cdot \erm^{-\epsilon}\cdot \p[f(P',k)=W] \tag*{($\rho$-PJR)} \\
    &\geqslant \rho^2\cdot \erm^{-2\epsilon}\cdot \p[f(P,k)=W], \tag*{($\epsilon$-DP)}
  \end{align*}
  which indicates that $\rho^2\cdot \erm^{-2\epsilon}\leqslant 1$, i.e., $\rho\leqslant \erm^{\epsilon}$. That completes the proof.
\end{proofwithname}

\begin{proofwithname}[Proof of Theorem \ref{thm: 2-EJR-upper}]
  Suppose $f\colon \themap$ be an ABC rule satisfying $\epsilon$-DP and $\kappa$-EJR.
  Consider the following voting instances $(P,k)$ and $(P',k)$, where $s=\lceil n/k\rceil$:
  \begin{align*}
    P_j &= \begin{cases}
      \{a_1\}, & 1\leqslant j\leqslant s, \\
      \{a_2\}, & s+1\leqslant j\leqslant 2s \\
      ~\cdots \\
      \{a_k\}, & (k-1)s+1\leqslant j\leqslant m.
    \end{cases} \\
    P_j' &= \begin{cases}
      \{a_{k+1}\}, & 1\leqslant j\leqslant s, \\
      \{a_2\}, & s+1\leqslant j\leqslant 2s \\
      ~\cdots \\
      \{a_k\}, & (k-1)s+1\leqslant j\leqslant m.
    \end{cases}
  \end{align*}
  Without loss of generality, suppose $n=s\cdot k$. Then for each $t\in\{0,1,\ldots,m-1\}$, the set of voters $\{ts+1,ts+2,\ldots, (t+1)s\}$ forms a $1$-cohesive group in both $P$ and $P'$. Therefore, $W=\{a_1,a_2,\ldots,a_k\}$ and $W'=\{a_{k+1},a_2,a_3,\ldots,a_k\}$ are the unique JR committee for $(P,k)$ and $(P',k)$, respectively. Noting that there is no $\ell$-cohesive group ($\ell\geqslant 2$) for either $P$ or $P'$, $W$ and $W'$ are the unique EJR committee for $(P,k)$ and $(P',k)$, respectively. Further, since $f$ satisfies $\kappa$-EJR, we have
  \begin{align}
    \label{equ: proof-EJR-kappa}
    \begin{aligned}
        \p[f(P,k)=W] &\geqslant \kappa\cdot \p[f(P,k)=W'], \text{and} \\
        \p[f(P',k)=W'] &\geqslant \kappa\cdot \p[f(P',k)=W].
    \end{aligned}
  \end{align}
  On the other hand, $P$ and $P'$ differ on $s$ voters' vote. By $\epsilon$-DP, for all $X\in\Ak$
  \begin{align}
    \label{equ: proof-EJR-epsilon}
    \begin{aligned}
        \erm^{-s\epsilon}\cdot \p[f(P',k)=X] &\leqslant \p[f(P,k)=X] \\
    &\leqslant \erm^{s\epsilon}\cdot \p[f(P',k)=X].
    \end{aligned}
  \end{align}
  Further, we have
  \begin{align*}
    \p[f(P,k)=W] &\geqslant \kappa\cdot \p[f(P,k)=W'] \tag*{(By Equation (\ref{equ: proof-EJR-kappa}))} \\
    &\geqslant \kappa\cdot \erm^{-\lceil\frac{n}{k}\rceil\cdot\epsilon}\cdot \p[f(P',k)=W'] \tag*{(By Equation (\ref{equ: proof-EJR-epsilon}))} \\
    &\geqslant \kappa^2\cdot \erm^{-\lceil\frac{n}{k}\rceil\cdot\epsilon}\cdot \p[f(P',k)=W] \tag*{(By Equation (\ref{equ: proof-EJR-kappa}))} \\
    &\geqslant \kappa^2\cdot \erm^{-2\lceil\frac{n}{k}\rceil\cdot\epsilon}\cdot \p[f(P,k)=W]. \tag*{(By Equation (\ref{equ: proof-EJR-epsilon}))}
  \end{align*}
  Therefore, $\kappa^2\cdot \erm^{-2\lceil\frac{n}{k}\rceil\cdot\epsilon}\leqslant 1$, i.e., $\kappa\leqslant \erm^{\lceil\frac{n}{k}\rceil\cdot\epsilon}$, which completes the proof.
\end{proofwithname}

\begin{proofwithname}[Proof of Proposition \ref{prop: 2-JR-lower}]
  We only prove for JR here, proofs for the cases when $\mathcal{A}=\PJR/\EJR$ can be obtained in the same way. Let $\mathfrak{R}_{\JR}\colon \themap$ denote the ABC rule induced by Mechanism \ref{algo: RR-JR} with axiom $\mathcal{A}=\JR$ and noise level $\epsilon\in\R_+$. For any voting instance $(P,k)$, supposing $t=|\JR(P,k)|$, the winning probability of each committee $W\in\Ak$ satisfies
  \begin{align*}
    \p[f(P,k)=W] = \begin{cases}
      \frac{\erm^{\epsilon/2}}{t\cdot\erm^{\epsilon/2} + \binom{m}{k}-t}, & W\in\JR(P,k), \\
      \frac{1}{t\cdot\erm^{\epsilon/2} + \binom{m}{k}-t}, & \text{otherwise}.
    \end{cases}
  \end{align*}
  Then, for any neighboring profile $P'$ of $P$ (let $t'=|\JR(P',k)|$) and committee $W\in\Ak$,
  \begin{align*}
    \frac{\p[f(P,k)=W]}{\p[f(P',k)=W]} &\leqslant \frac{\erm^{\epsilon/2}}{t\cdot\erm^{\epsilon/2} + \binom{m}{k}-t} / \frac{1}{t'\cdot\erm^{\epsilon/2} + \binom{m}{k}-t'} \\
    &= \erm^{\epsilon/2}\cdot \frac{t'\cdot\erm^{\epsilon/2} + \binom{m}{k}-t'}{t\cdot\erm^{\epsilon/2} + \binom{m}{k}-t} \\
    &\leqslant \erm^{\epsilon/2}\cdot \frac{\erm^{\epsilon/2}\cdot \binom{m}{k}}{\binom{m}{k}} \\
    &= \erm^\epsilon,
  \end{align*}
  which indicates that $\mathfrak{R}_{\JR}$ satisfies $\epsilon$-DP. Further, according to the definition of Mechanism \ref{algo: RR-JR}, it is easy to check that for any voting instance $(P,k)$ and committees $W\in\JR(P,k)$ and $W'\notin \JR(P,k)$,
  \begin{align*}
    \frac{\p[f(P,k)=W]}{\p[f(P,k)=W']} = \frac{\erm^{\chi(W)\cdot\epsilon/2}}{\erm^{\chi(W)\cdot\epsilon/2}} = \erm^{\epsilon/2},
  \end{align*}
  i.e., $\mathfrak{R}_{\JR}$ satisfies $\erm^{\epsilon/2}$-JR, which completes the proof.
\end{proofwithname}

\subsection{Missing proofs in subsection \ref{subsec: 2-eff}}

\begin{proofwithname}[Proof of Theorem \ref{thm: 2-PE-upper}]
    Let $f\colon\themap$ be an ABC rule satisfying $\epsilon$-DP and $\beta$-PE. Consider the following voting instance $E=(P,k)$, where $m\geqslant n+2k-1$ and $\{a_1,\ldots,a_k,b_1,\ldots,\linebreak b_{n-1},c_1,\ldots,c_k\}\subseteq A$.
    \begin{align*}
      P_j = \begin{cases}
        \{a_1,a_2,\ldots,a_k\}, & j=1, \\
        \{a_1,a_2,\ldots,a_k,b_1\}, & j=2, \\
        ~\cdots \\
        \{a_1,a_2,\ldots,a_k,b_1,b_2,\ldots,b_{n-1}\}, & j=n. \\
      \end{cases}
    \end{align*}
    Then we can construct the following committees.
    \begin{table}[H]
      \renewcommand\arraystretch{1.25}
      \centering
      \begin{tabular}{@{}l@{~}l@{~}l@{~}l@{}}
        $\{a_1,\ldots,a_k\}$ & $\{a_1,\ldots,a_{k-1},b_1\}$ & $\cdots$ & $\{a_1,\ldots,a_{k-1},b_{n-1}\}$ \\
        $\{a_1,\ldots,a_{k-1},c_1\}$ & $\{a_1,\ldots,a_{k-2},b_1,c_1\}$ & $\cdots$ & $\{a_1,\ldots,a_{k-2},b_{n-1},c_1\}$ \\
        $\vdots$ & $\vdots$ &  & $\vdots$ \\
        $\{a_1,c_1,\ldots,c_{k-1}\}$ & $\{b_1,c_1,\ldots,c_{k-1}\}$ & $\cdots$ & $\{b_{n-1},c_1,\ldots,c_{k-1}\}$ \\
        $\{c_1,\ldots,c_k\}$ &&&
      \end{tabular}
    \end{table}
    To be more concise, we use $W_{p,q}$ to denote the committee on $p$-th row and $q$-th column in the matrix above. Then for all $1\leqslant p\leqslant k$ and $1\leqslant q\leqslant n$, we have
    \begin{align*}
      W_{p,q} = \begin{cases}
        \{a_1,\ldots,a_{k-p+1},c_1,\ldots,c_{p-1}\}, & q=1, \\
        \{a_1,\ldots,a_{k-p},b_{q-1},c_1,\ldots,c_{p-1}\}, & q>1.
      \end{cases}
    \end{align*}
    For the sake of simplicity, we use $\Omega_P(W)$ to denote the series consisting of $|W\cap P_j|$ for each voter $j\in N$, i.e., $\Omega_P\colon \Ak\to\mathbb{N}^n,~W\mapsto (|W\cap P_1|,|W\cap P_2|,\ldots,|W\cap P_n|)$. Then it is easy to check that $\Omega_P(W_{p,q})$ forms the Table \ref{tab: OmegaW}.
    \begin{table*}[h]
        \caption{\boldmath Table of $\Omega(W_{p,q})$}
        \renewcommand\arraystretch{1.25}
        \centering
        \begin{tabular}{cccc}
            $(k,k,\ldots,k)$ & $(k-1,k,\ldots,k)$ & $\cdots$ & $(k-1,\ldots,k-1,k)$ \\
            $(k-1,k-1,\ldots,k-1)$ & $(k-2,k-1,\ldots,k-1)$ & $\cdots$ & $(k-2,\ldots,k-2,k-1)$ \\
            $\vdots$ & $\vdots$ &  & $\vdots$ \\
            $(1,1,\ldots,1)$ & $(0,1,\ldots,1)$ & $\cdots$ & $(0,\ldots,0,1)$ \\
            $(0,0,\ldots,0)$ &&&
        \end{tabular}
        \label{tab: OmegaW}
    \end{table*}
    In other words, we have
    \begin{align*}
      \Omega_P(W_{p,q}) = (\overbrace{k-p,\ldots,k-p}^{q-1},\underbrace{k-p+1,\ldots,k-p+1}_{n-q+1}).
    \end{align*}
    Then, according to the definition of Pareto-dominance, we have $W_{p_1,q_1}$ Pareto dominates $W_{p_2,q_2}$ if $p_1<p_2$ or $p_1=p_2$ and $q_1<q_2$. The relationship of Pareto-dominance among $W_{p,q}$ can be visualized as the following diagram, where $X\to Y$ indicates that $X$ Pareto-dominates $Y$.
    \begin{align*}
      \begin{matrix}
        W_{1,1} & \longrightarrow & W_{1,2} & \longrightarrow & \cdots & \longrightarrow & W_{1,n} \\
        &&&&&& \downarrow \\
        W_{2,n} & \longleftarrow & W_{2,n-1} & \longleftarrow & \cdots & \longleftarrow & W_{2,1} \\
        \downarrow \\
        \vdots \\
        \downarrow \\
        W_{k,1} & \longrightarrow & W_{k,2} & \longrightarrow & \cdots & \longrightarrow & W_{k,n} & \longrightarrow & W_{k+1,1}
      \end{matrix}
    \end{align*}
    There is totally $nk$ arrows in the diagram. Since $f$ satisfies $\beta$-PE, we have
    \begin{align*}
      \p[f(P,k)=W_{1,1}] \geqslant \beta^{nk}\cdot \p[f(P,k)=W_{k+1,1}].
    \end{align*}
    However, since $f$ satisfies $\epsilon$-DP, Lemma \ref{lem: dp-neps} indicates that
    \begin{align*}
      \p[f(P,k)=W_{1,1}] \leqslant \erm^{n\epsilon}\cdot \p[f(P,k)=W_{k+1,1}].
    \end{align*}
    In other words, we have $\beta^{nk}\leqslant \erm^{n\epsilon}$, i.e., $\beta\leqslant \erm^{\epsilon/k}$, which completes the proof.
  \end{proofwithname}

  \begin{proofwithname}[Proof of Proposition \ref{prop: 2-PE-lower}]
    Let $f$ denote the ABC rule corresponding to Mechanism \ref{algo: AVExp}. Then, for any committee $W\in\Ak$, we have
    \begin{align*}
      \p[f(P,k)=W] &\propto \prod_{a\in W} \erm^{\chi(a)\cdot\epsilon/(2k)} \\
      &= \erm^{\sum\nolimits_{\chi(a)}\cdot\epsilon/(2k)} \\
      &= \erm^{{\rm AV}_P(W)\cdot\epsilon/(2k)}.
    \end{align*}
    First, we prove $f$ satisfies $\epsilon$-DP. Suppose $P$ and $P'$ are a pair of neighboring profiles, satisfying $P_i\neq P_i'$ and $P_{-i} = P_{-i}'$. Then the scores of any committee $W$ in $P$ and $P'$ satisfy
    \begin{align*}
        {\rm AV}_P(W) - {\rm AV}_{P'} \leqslant {\rm AV}_{P_i}(W)-{\rm AV}_{P_i'}(W) \leqslant k.
    \end{align*}
    Therefore, we have
    \begin{align*}
        \frac{\p[f(P,k)=W]}{\p[f(P',k)=W]} = \frac{\erm^{{\rm AV}_P(W)\cdot \epsilon/(2k)}}{\erm^{{\rm AV}_{P'}(W)\cdot \epsilon/(2k)}} \leqslant \erm^{\epsilon},
    \end{align*}
    which indicates that $f$ satisfies $\epsilon$-DP.
  
    Next, we prove the Pareto bound of $f$. Suppose that in profile $P$, committee $W$ Pareto dominates $W'$. Then we have
    \begin{itemize}
        \item for all $j\in N$, $|P_j\cap W| - |P_j\cap W'|\geqslant 0$;
        \item for at least one $j\in N$, $|P_j\cap W| - |P_j\cap W'| \geqslant 1$.
    \end{itemize}
    Therefore, their scores satisfy
    \begin{align*}
        {\rm AV}_P(W) - {\rm AV}_P(W') = \sum\limits_{j\in N} {\rm AV}_{P_j}(W)-{\rm AV}_{P_j}(W') \geqslant 1.
    \end{align*}
    Then we have
    \begin{align*}
        \frac{\p[f(P,k)=W]}{\p[f(P,k)=W']} \geqslant \erm^{\frac{\epsilon}{2k}},
    \end{align*}
    which indicates that $f$ satisfies $\erm^{\frac{\epsilon}{2k}}$-Pareto efficiency. That completes the proof.
  \end{proofwithname}

  \begin{proofwithname}[Proof of Theorem \ref{thm: 2-CC-upper}]
    Let $f\colon\themap$ be an ABC rule satisfying $\epsilon$-DP and $\eta$-CC.
    Consdier the voting instances $(P,k)$ and $(P',k)$, where $n=2t+1$ and there exists a $W\in\powerset{A,k-1}$ that
    \begin{align*}
      P_j &= \begin{cases}
        W\cup \{a_1\}, & 1\leqslant j\leqslant t+1, \\
        W\cup \{a_2\}, & t+2\leqslant j\leqslant 2t+1.
      \end{cases} \\
      P_j' &= \begin{cases}
        W\cup \{a_1\}, & 1\leqslant j\leqslant t, \\
        W\cup \{a_2\}, & t+1\leqslant j\leqslant 2t+1.
      \end{cases}
    \end{align*}
    Then $P$ and $P'$ are neighboring profiles, since they only differ on the $(t+1)$-th voter's vote. Further, according to the definition, $W\cup \{a_1\}$ is the Condorcet committee in $(P,k)$, since it is approved by $t+1$ out of $2t+1$ voters (more than a half). Similarly, $W\cup \{a_2\}$ is the Condorcet committee in $(P',k)$. Further, we have
    \begin{align*}
      \p[f(P,k)=W\cup \{a_1\}] \geqslant&\; \eta\cdot \p[f(P,k)=W\cup \{a_2\}] \\
      \geqslant&\; \erm^{-\epsilon}\cdot \eta\cdot \p[f(P',k)=W\cup \{a_2\}] \\
      \geqslant&\; \erm^{-\epsilon}\cdot \eta^2\cdot \p[f(P',k)=W\cup \{a_1\}] \\
      \geqslant&\; \erm^{-2\epsilon}\cdot \eta^2\cdot \p[f(P,k)=W\cup \{a_1\}].
    \end{align*}
    In other words, $\erm^{-2\epsilon}\cdot \eta^2\leqslant 1$, i.e., $\eta\leqslant \erm^{\epsilon}$, which completes the proof.
  \end{proofwithname}

  \begin{proofwithname}[Proof of Proposition \ref{prop: 2-CC-lower}]
    Let $f$ denote the ABC rule corresponding to Mechanism \ref{algo: Condorcet-RR}. First, we prove that $f$ satisfies $\epsilon$-DP. In fact, given any neighboring voting instances $(P,k)$ and $(P',k)$, for each committee $W\in\Ak$, we have
    \begin{align*}
        \frac{\p[f(P,k)=W]}{\p[f(P',k)=W]} \leqslant \frac{\sup_W \p[f(P,k)=W]}{\inf_W \p[f(P,k)=W]} = \erm^{\epsilon},
    \end{align*}
    which indicates that $f$ satisfies $\epsilon$-DP.

    Next, we prove that $f$ satisfies $\erm^\epsilon$-Condorcet criterion. Let $P\in\powerset{A}^n$ be a profile, where exists a Condorcet committee $W_c$. Then for any $W\in\Ak\backslash \{W_c\}$,
    \begin{align*}
        \frac{\p[f(P,k)=W_c]}{\p[f(P,k)=W]} = \frac{\erm^{\epsilon}}{\erm^{\epsilon}+\binom{m}{k}-1} / \frac{1}{\erm^{\epsilon}+\binom{m}{k}-1} = \erm^\epsilon.
    \end{align*}
    In other words, $f$ satisfies $\erm^{\epsilon}$-Condorcet criterion, which completes the proof.
\end{proofwithname}

\section{Missing Proofs in Section \ref{sec: 3-tradeoff}}
\label{sec: appendix-B}

\subsection{The compatibility between axioms without DP}
\label{sec: appendix-B1}

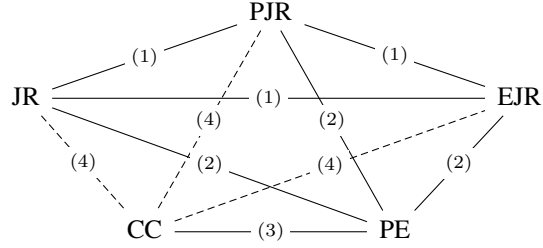
\begin{figure}[H]
  \centering
  \begin{tikzcd}
    && {\text{PJR}} \\
    {\text{JR}} &&&& {\text{EJR}} \\
    \\
    & {\text{CC}} && {\text{PE}}
    \arrow["{(1)}"{description}, no head, from=1-3, to=2-5]
    \arrow["{(4)}"{description}, dashed, no head, from=1-3, to=4-2]
    \arrow["{(2)}"{description}, no head, from=1-3, to=4-4]
    \arrow["{(1)}"{description}, no head, from=2-1, to=1-3]
    \arrow["{(1)}"{description}, no head, from=2-1, to=2-5]
    \arrow["{(4)}"{description}, dashed, no head, from=2-1, to=4-2]
    \arrow["{(2)}"{description}, no head, from=2-1, to=4-4]
    \arrow["{(2)}"{description}, no head, from=2-5, to=4-4]
    \arrow["{(4)}"{description}, dashed, no head, from=4-2, to=2-5]
    \arrow["{(3)}"{description}, no head, from=4-2, to=4-4]
  \end{tikzcd}
  \caption{Compatibility between axioms, where a solid line indicates that the axioms are compatible, while a dash line indicates that the axioms are incompatible.}
\end{figure}

First of all, the compatibility among JR-family (denoted by (1)) is evident, since there are implication relationships among them. The correctness of (2) is ensured by the proportional approval voting (PAV) rule, which satisfies PE and EJR \citep{aziz2017justified}. Further, (3) is correct since the Condorcet committee is never Pareto-dominated. Finally, CC and JR are incompatible (denoted by (4)) since the Condorcet committee does not necessarily satisfy JR. For example, consider the following voting instance $(P,k)$, where $n=2t+1$.
\begin{align*}
  P_j = \begin{cases}
    \{a_1,a_2,\ldots,a_k\}, & 1\leqslant j\leqslant t+1, \\
    \{a_{k+1},a_{k+2},\ldots,a_{2k}\}, & t+2\leqslant 2t+1.
  \end{cases}
\end{align*}
When $k\geqslant 3$, both $\{1,2,\ldots,t+1\}$ and $\{t+2,t+3,\ldots,2t+1\}$ are $1$-cohesive for $(P,k)$. However, the Condorcet committee of this voting instance is $\{a_1,a_2,\ldots,a_k\}$, which does not offer any seat to $\{t+2,t+3,\ldots,2t+1\}$. Therefore, CC is incompatible with JR (and PJR/EJR).

\subsection{The missing proofs}

\begin{proofwithname}[Proof of Theorem \ref{thm: 3-JR-PJR-upper}]
  First, we prove (1). Let $f\colon\themap$ be an ABC rule satisfying $\epsilon$-DP, $\theta$-JR, and $\rho$-PJR. Consider the following neighboring voting instances $(P,k)$ and $(P',k)$, where $s=\lceil 2n/k\rceil$:
  \begin{align*}
    P_j &= \begin{cases}
      \{a_1,a_2\}, & 1\leqslant j\leqslant s, \\
      \{a_1,a_3\}, & s+1\leqslant j\leqslant 2s-1, \\
      \{a_4,a_5\}, & 2s\leqslant j\leqslant 3s-1, \\
      \{a_4\}, & \text{otherwise}.
    \end{cases} \\
    P_j' &= \begin{cases}
      \{a_1,a_2\}, & 1\leqslant j\leqslant s-1, \\
      \{a_1,a_3\}, & s\leqslant j\leqslant 2s-1, \\
      \{a_4,a_5\}, & 2s\leqslant j\leqslant 3s-1, \\
      \{a_4\}, & \text{otherwise}.
    \end{cases}
  \end{align*}
  
  By the definition of cohesive groups, it is easy to check that $V=\{1,2,\ldots,s\}$ and $V'=\{s,s+1,\ldots,2s-1\}$ are $2$-cohesive for $(P,k)$ and $(P',k)$, respectively. For both $P$ and $P'$, $\{2s,2s+1,\ldots,3s-1\}$ forms a $2$-cohesive group, while $\{1,2,\ldots,2s-1\}$ and $\{2s,2s+1,\ldots,n\}$ are $1$-cohesive. Further, there is no $\ell$-cohesive group with $\ell\geqslant 3$ for both voting instances.
  
  Therefore, $W_1=\{a_1,a_2,a_4,a_5,a_6,\ldots,a_{k+1}\}$ satisfies PJR for $(P,k)$, but does not satisfy even JR for $(P',k)$, since it offers only one seat for the $2$-cohesive group $V'$. Similarly, $W_1'=\{a_1,a_3,a_4,a_5,a_6,\linebreak \ldots,a_{k+1}\}$ satisfies PJR for $(P',k)$, but does not satisfy JR for $(P,k)$. Besides, $W_0=\{a_1,a_4,a_6,a_7,\ldots,a_{k+3}\}$ satisfies JR, but does not satisfy PJR for both voting instances.

  Since $f$ satisfies $\epsilon$-DP, $\theta$-JR, and $\rho$-PJR, we have
  \begin{align*}
    \p[f(P,k)=W_1] &\geqslant \rho\cdot \p[f(P,k)=W_0] \tag*{($\rho$-PJR)} \\
    &\geqslant \rho\cdot\theta\cdot \p[f(P,k)=W_1'] \tag*{($\theta$-JR)} \\
    &\geqslant \rho\cdot\theta\cdot\erm^{-\epsilon}\cdot \p[f(P',k)=W_1'] \tag*{($\epsilon$-DP)} \\
    &\geqslant \rho^2\cdot\theta\cdot\erm^{-\epsilon}\cdot \p[f(P',k)=W_0] \tag*{($\rho$-PJR)} \\
    &\geqslant \rho^2\cdot\theta^2\cdot\erm^{-\epsilon}\cdot \p[f(P',k)=W_1] \tag*{($\theta$-JR)} \\
    &\geqslant \rho^2\cdot\theta^2\cdot\erm^{-2\epsilon}\cdot \p[f(P,k)=W_1] \tag*{($\epsilon$-DP)}.
  \end{align*}
  In other words, we have $\rho^2\cdot\theta^2\cdot\erm^{-2\epsilon}\leqslant 1$, i.e., $\rho\cdot\theta\leqslant \erm^\epsilon$, which completes the proof of (1).

  For (2), we still consider the two voting instances defined above. In fact, $W_1$ and $W_1'$ satisfy EJR for $(P,k)$ and $(P',k)$, respectively. Therefore, by comparing the winning probability of $W_0$, $W_1$, and $W_1'$, we will also end up with 
  \begin{align*}
    \p[f(P,k)=W_1] \leqslant \kappa^2\cdot\theta^2\cdot\erm^{-2\epsilon}\cdot \p[f(P,k)=W_1].
  \end{align*}
  In other words, we have $\kappa\cdot \theta\leqslant \erm^\epsilon$, which completes the proof.
\end{proofwithname}

\begin{proofwithname}[Proof of Theorem \ref{thm: 3-PJR-EJR-upper}]
  Let $f\colon\themap$ be an ABC rule satisfying $\epsilon$-DP, $\rho$-PJR, and $\kappa$-EJR. Consider the following voting instance $(P,k)$, where $s=\lceil n/k\rceil$:
  \begin{align*}
    P_j = \begin{cases}
      \{a_1,a_2,\ldots,a_k,a_{k+1}\}, & 1\leqslant j\leqslant s, \\
      \{a_1,a_2,\ldots,a_k,a_{k+2}\}, & s+1\leqslant j\leqslant 2s, \\
      ~\cdots \\
      \{a_1,a_2,\ldots,a_k,a_{2k}\}, & (k-1)s+1\leqslant j\leqslant n.
    \end{cases}
  \end{align*}
  Notice that $\{a_1,a_2,\ldots,a_k\}$ is approved by every voter $j\in N$. As a consequence, for all $\ell\in\{1,2,\ldots,k\}$, the set of voters $V\subseteq N$ forms an $\ell$-cohesive group if and only if $|V|\geqslant \ell\cdot n/k$. Further, it is not hard to check correctness the following statements for $(P,k)$.
  \begin{enumerate}
    \item $W_1=\{a_1,a_2,\ldots,a_k\}$ satisfies EJR;
    \item $W_0=\{a_{k+1},a_{k+2},\ldots,a_{2k}\}$ satisfies PJR, but fails to satisfy EJR.
  \end{enumerate}
  Now, consider another voting instance $(P',k)$, where
  \begin{align*}
    P_j' = \begin{cases}
      \{a_{k+1},a_{2k+1},a_{2k+2},\ldots,a_{3k-1}\}, & 1\leqslant j\leqslant s, \\
      P_j, & \text{otherwise}.
    \end{cases}
  \end{align*}
  Then $\{1,2,\ldots,s\}$ forms a $1$-cohesive group, and any $V\subseteq \{s+1,s+2,\ldots,n\}$ is $\ell$-cohesive if and only if $|V|\geqslant \ell\cdot n/k$. Therefore, the following statements are true.
  \begin{enumerate}
    \item $W_1'=\{a_1,a_2,\ldots,a_{k-1},a_{2k+1}\}$ satisfies EJR for $(P',k)$;
    \item $W_0$ also satisfies PJR for $(P',k)$, but fails to satisfy EJR.
  \end{enumerate}
  Further, $W_1$ does not satisfy JR for $(P,k)$, since it does not offer any seat for the $1$-cohesive group $\{1,2,\ldots,s\}$. Similarly, $W_1'$ does not satsify JR for $(P,k)$. Notice that $P$ and $P'$ only differ on the first $s$ voter's vote. Therefore, for any $W\in\Ak$, 
  \begin{align*}
    \p[f(P,k)=W] \leqslant \erm^{s\epsilon}\cdot \p[f(P',k)=W].
  \end{align*}
  Since $f$ satisfies $\epsilon$-DP, $\rho$-PJR, and $\kappa$-EJR, we have
  \begin{align*}
    \p[f(P,k)=W_1] &\geqslant \kappa\cdot \p[f(P,k)=W_0] \tag*{($\kappa$-EJR)} \\
    &\geqslant \kappa\cdot\rho\cdot \p[f(P,k)=W_1'] \tag*{($\rho$-PJR)} \\
    &\geqslant \kappa\cdot\rho\cdot\erm^{-s\epsilon}\cdot \p[f(P',k)=W_1'] \tag*{($\epsilon$-DP)} \\
    &\geqslant \kappa^2\cdot\rho\cdot\erm^{-s\epsilon}\cdot \p[f(P',k)=W_0] \tag*{($\kappa$-EJR)} \\
    &\geqslant \kappa^2\cdot\rho^2\cdot\erm^{-s\epsilon}\cdot \p[f(P',k)=W_1] \tag*{($\rho$-PJR)} \\
    &\geqslant \kappa^2\cdot\rho^2\cdot\erm^{-2s\epsilon}\cdot \p[f(P,k)=W_1]. \tag*{($\epsilon$-DP)}
  \end{align*}
  In other words, we have $\kappa\cdot \rho\leqslant \erm^{\lceil n/k \rceil\cdot\epsilon}$, which completes the proof.
\end{proofwithname}

\begin{proofwithname}[Proof of Theorem \ref{thm: 3-PE-JR-upper}]
  First, we prove (1). Let $f\colon\themap$ be a neutral ABC rule satisfying $\epsilon$-DP, $\beta$-PE, and $\theta$-JR. Consider the following voting instance $E=(P,k)$, where $m\geqslant n+2k-1$ and $\{a_1,\ldots,a_k,b_1,\ldots,b_{n-1},c_1,\ldots,c_k\}\subseteq A$.
  \begin{align*}
    P_j = \begin{cases}
      \{a_1,a_2,\ldots,a_k\}, & j=1, \\
      \{a_1,a_2,\ldots,a_k,b_1\}, & j=2, \\
      ~\cdots \\
      \{a_1,a_2,\ldots,a_k,b_1,b_2,\ldots,b_{n-1}\}, & j=n. \\
    \end{cases}
  \end{align*}
  Then we can construct the following committees, where we use $W_{p,q}$ to denote the committee on $p$-th row and $q$-th column in the table.
  \begin{table}[H]
    \renewcommand\arraystretch{1.25}
    \centering
    \begin{tabular}{@{}l@{~}l@{~}l@{~}l@{}}
      $\{a_1,\ldots,a_k\}$ & $\{a_1,\ldots,a_{k-1},b_1\}$ & $\cdots$ & $\{a_1,\ldots,a_{k-1},b_{n-1}\}$ \\
      $\{a_1,\ldots,a_{k-1},c_1\}$ & $\{a_1,\ldots,a_{k-2},b_1,c_1\}$ & $\cdots$ & $\{a_1,\ldots,a_{k-2},b_{n-1},c_1\}$ \\
      $\vdots$ & $\vdots$ &  & $\vdots$ \\
      $\{a_1,c_1,\ldots,c_{k-1}\}$ & $\{b_1,c_1,\ldots,c_{k-1}\}$ & $\cdots$ & $\{b_{n-1},c_1,\ldots,c_{k-1}\}$ \\
      $\{c_1,\ldots,c_k\}$ &&&
    \end{tabular}
  \end{table}
  Notice that the voting instance and the table of committees are exactly the same as those in the proof of Theorem \ref{thm: 2-PE-upper}. Then the relationship of Pareto-dominance among $W_{p,q}$ can be visualized as the same diagram, where the arrows are labeled from $1$ to $kn$.
  \begin{align*}
    \begin{matrix}
      W_{1,1} & \xrightarrow{~~1~~} & W_{1,2} & \xrightarrow{~~2~~} & \cdots & \xrightarrow{n-1} & W_{1,n} \\
      &&&&&& \rotatebox[origin=c]{270}{$\xrightarrow{n}$} \\
      W_{2,n} & \xleftarrow{2n-1} & W_{2,n-1} & \xleftarrow{2n-2} & \cdots & \xleftarrow{n+1} & W_{2,1} \\
      \rotatebox[origin=c]{270}{$\xrightarrow{2n}$} \\
      \vdots \\
      \rotatebox[origin=c]{270}{$\xrightarrow{(k-1)n}$} \\
      W_{k,1} & \xrightarrow{(k-1)n+1} & W_{k,2} & \xrightarrow{(k-1)n+2} & \cdots & \xrightarrow{kn-1} & W_{k,n} \\
      &&&&&& \rotatebox[origin=c]{270}{$\xrightarrow{kn}$} \\
      &&&&&& W_{k+1,1} \\
    \end{matrix}
  \end{align*}
  Since $a_1,a_2,\ldots,a_k$ are approved by all voters, the group of voters $V\subseteq N$ is $\ell$-cohesive if and only if $|V|\geqslant \ell\cdot n/k$, for any $\ell$.
  Further, according to the definition of JR, it is easy to check that
  \begin{align*}
    W_{1,1}\in\JR(P,k)\quad\text{and}\quad W_{k+1,1}\notin\JR(P,k).
  \end{align*}
  Therefore, there must exist an arrow $\xrightarrow{~~t~~}$ in the diagram, connecting a JR committee and a non-JR committee, i.e., there exists a pair of adjacent committees $X,Y$ in the diagram, satisfying $X\in\JR(P,k)$ and $Y\notin \JR(P,k)$. Then we have
  \begin{align*}
    \p[f(P,k)=X] \geqslant \theta\cdot \p[f(P,k)=Y].
  \end{align*}
  Since there are totally $t-1$ arrows between $W_{1,1}$ and $X$, we have
  \begin{align*}
    \p[f(P,k)=W_{1,1}] \geqslant \beta^{t-1}\cdot \p[f(P,k)=X].
  \end{align*}
  Similarly, there are $kn-t$ arrows between $Y$ and $W_{k+1,1}$, which indicates that
  \begin{align*}
    \p[f(P,k)=Y] \geqslant \beta^{kn-t}\cdot \p[f(P,k)=W_{k+1,1}].
  \end{align*}
  According to the three inequalities above, we have
  \begin{align}
    \label{equ: 3-PE-JR-proof}
    \p[f(P,k)=W_{1,1}] \geqslant \beta^{kn-1}\cdot\theta\cdot \p[f(P,k)=W_{k+1,1}].
  \end{align}
  However, since $f$ satisfies $\epsilon$-DP, Lemma \ref{lem: dp-neps} indicates that
  \begin{align*}
    \p[f(P,k)=W_{1,1}] \leqslant \erm^{n\epsilon}\cdot \p[f(P,k)=W_{k+1,1}].
  \end{align*}
  In other words, we have $\beta^{nk-1}\cdot\theta\leqslant \erm^{n\epsilon}$, which completes the proof of (1).
  For (2) and (3), we still consider the voting instances and committees defined above. Then it is easy to check that $W_{1,1}$ also satisfies PJR and EJR for $(P,k)$. Therefore, the inequality (\ref{equ: 3-PE-JR-proof}) also holds when $\theta$-JR is replaced by $\rho$-PJR or $\kappa$-EJR. That completes the proof.
\end{proofwithname}

\begin{proofwithname}[Proof of Theorem \ref{thm: 3-PE-CC-upper}]
  Let $f\colon\themap$ be a neutral ABC rule satisfying $\epsilon$-DP, $\beta$-PE, and $\eta$-CC. We still consider the voting instance $(P,k)$ and the committees $\{W_{p,q}\}$ defined in the proof of Theorem \ref{thm: 2-PE-upper}. Then it is easy to check that the Condorcet committee for $(P,k)$ is $W_{1,1}$. Therefore, 
  \begin{align*}
    \p[f(P,k)=W_{1,1}] &\geqslant \eta\cdot \p[f(P,k)=W_{1,2}] \\
    &\geqslant \eta\cdot\beta^{nk-1}\cdot \p[f(P,k)=W_{k+1,1}].
  \end{align*}
  Then, by Lemma \ref{lem: dp-neps}, we have $\eta\cdot\beta^{nk-1}\leqslant \erm^{n\epsilon}$, which completes the proof.
\end{proofwithname}

\end{document}